\documentclass[10pt,twocolumn,twoside]{IEEEtran}

\usepackage[active]{srcltx} %SRC Specials for DVI Searching
\usepackage[cmex10]{amsmath}
\usepackage{amsfonts}
\usepackage[final]{graphics}
\usepackage[final]{graphicx}
\usepackage{amsbsy}
\usepackage{amsbsy}
\usepackage{amssymb}
\usepackage{url}
\usepackage{enumerate}
\usepackage{cite}
\usepackage{psfrag}
\usepackage{color}
\usepackage{euscript}
\usepackage[utf8]{inputenc}
\usepackage{soul}
\usepackage{tikz}
\usepackage{textcomp}
\usepackage{lipsum}

\newcommand{\bm}[1]{{\boldsymbol{#1}}}
\newcommand{\range}{\mathcal{R}}
\newcommand{\Orange}{\mathcal{R}^{\bot}}
\newcommand{\nullo}{\mathcal{N}}

\newcommand{\diag}{{\text{diag}}}

\newcommand{\trace}{{\text{trace}}}
\renewcommand{\vec}{{\text{vec}}}

\newcommand{\rank}{{\text{rank}}}

\newcommand{\Es}{{\mathbb{E}}}

\newcommand{\Id}{\mathbf{I}}
\newcommand{\Zero}{\mathbf{O}}

\renewcommand{\d}{\mathbf{d}}

\newcommand{\A}{\mathbf{A}}
\newcommand{\Fcal}{\bm{ \mathcal B}}

\newcommand{\C}{\mathbb{C}}

\newcommand{\N}{\mathbb{N}}

\newcommand{\eqdef}{\triangleq}
\newcommand{\ba}{\mathbf{a}}
\newcommand{\bA}{\mathbf{A}}

\newcommand{\bB}{\mathbf{B}}

\newcommand{\bd}{\mathbf{d}}

\newcommand{\bh}{\mathbf{h}}

\newcommand{\bhb}{\mathbf{h}_{\text{B}}}
\newcommand{\bhbm}{\mathbf{h}_{\text{B},m}}

\newcommand{\bhbhatmmseuno}{\widehat{\mathbf{h}}_{\text{B},\text{MMSE}}}
\newcommand{\bhbhatmmsesub}{\widehat{\mathbf{h}}_{\text{B},\text{MMSE-SUB}}}
\newcommand{\bhbhatmmsesmi}{\widehat{\mathbf{h}}_{\text{B},\text{MMSE-SMI}}}
\newcommand{\bhbhatlmmseuno}{\widehat{\mathbf{h}}^{(1)}_{\text{B},\text{LMMSE}}}
\newcommand{\bhbhatlmmsedue}{\widehat{\mathbf{h}}^{(2)}_{\text{B},\text{LMMSE}}}

\newcommand{\bhbhatmmse}{\widehat{\mathbf{h}}_{\text{B},\text{MMSE}}}

\newcommand{\bH}{\mathbf{H}}
\newcommand{\bK}{\mathbf{K}}
\newcommand{\bI}{\mathbf{I}}

\newcommand{\bJ}{\mathbf{J}}
\newcommand{\bp}{\mathbf{p}}
\newcommand{\bP}{\mathbf{P}}

\newcommand{\bR}{\mathbf{R}}
\newcommand{\bO}{\mathbf{O}}

\newcommand{\bS}{\mathbf{S}}

\newcommand{\bU}{\mathbf{U}}
\newcommand{\bv}{\mathbf{v}}
\newcommand{\bV}{\mathbf{V}}
\newcommand{\bw}{\mathbf{w}}
\newcommand{\bx}{\mathbf{x}}

\newcommand{\bW}{\mathbf{W}}
\newcommand{\by}{\mathbf{y}}

\newcommand{\bY}{\mathbf{Y}}

\newcommand{\balpha}{\mbox{\boldmath $\alpha$}}

\newcommand{\bLambda}{\mbox{\boldmath $\Lambda$}}

\newcommand{\Nr}{{N}}

\newcommand{\Pa}{\EuScript{P}_{\text{B}}}
\newcommand{\Ps}{\EuScript{P}_{\text{E}}}
\newcommand{\Pamoverline}{\EuScript{P}_{\text{B},\overline{m}}}
\newcommand{\Pam}{\EuScript{P}_{\text{B},m}}
\newcommand{\Psm}{\EuScript{P}_{\text{E},m}}
\newcommand{\Psmoverline}{\EuScript{P}_{\text{E},\overline{m}}}
\newcommand{\Pshat}{\widehat{\EuScript{P}}_{\text{E}}}
\newcommand{\Psmax}{{\EuScript{P}}_{\text{E},\text{max}}}
\newcommand{\herm}{\text{H}}
\newcommand{\trasp}{\text{T}}
\def\Cap{\mathsf{C}}

\def\bdm#1\edm{\begin{displaymath}#1\end{displaymath}}
\def\be#1\ee{\begin{equation}#1\end{equation}}
\def\barr#1\earr{\begin{align}#1\end{align}}

\newcommand{\IeeeTIT}{{\em IEEE Trans.\ Inf. Theory\/}}
\newcommand{\IeeeTSP}{{\em IEEE Trans.\ Signal Process.\/}}
\newcommand{\IeeeTCOMM}{{\em IEEE Trans.\ Commun.\/}}

\newcommand{\IeeeWCOMMLETT}{{\em IEEE Wireless Commun.\ Lett.\/}}
\newcommand{\IeeeTWC}{{\em IEEE Trans.\ Wireless Commun.\/}}
\newcommand{\IeeeJSAC}{{\em IEEE J.\ Select.\ Areas Commun.\/}}
\newcommand{\IeeeTVT}{{\em IEEE Trans.\ Veh. Technol.\/}}

\newcommand{\IeeeTIFS}{{\em IEEE Trans. Inf. Foren. Sec.\/}}

\newtheorem{theorem}{Theorem}[section]

\newtheorem{lemma}[theorem]{Lemma}

\newcommand\acceptedtext{%
  \footnotesize This article has been accepted for publication in a future issue of this journal, but has not been fully edited. Content may change prior to final publication. \\
  Citation information: DOI 10.1109/TIFS.2020.2985548, IEEE Transactions on Information Forensics and Security.}
\newcommand\acceptednotice{%
\begin{tikzpicture}[remember picture,overlay]
\node[anchor=north,yshift=-6pt] at (current page.north) {%
\begin{minipage}{\textwidth}
\center \acceptedtext
\end{minipage}};
\end{tikzpicture}%
}

\begin{document}

\title{Design and performance analysis of channel estimators under pilot spoofing attacks
\\ in multiple-antenna systems}

\author{Donatella~Darsena,~\IEEEmembership{Senior Member,~IEEE},
             Giacinto~Gelli,~\IEEEmembership{Senior Member,~IEEE},
        Ivan~Iudice, \\ and Francesco~Verde,~\IEEEmembership{Senior Member,~IEEE}
\thanks{
Manuscript received December 26, 2019; revised March 28, 2020; 
accepted March 30, 2020. Date of publication XX YY, 2020; date of current version March 31, 2020. The associate editor coordinating the review of this manuscript and approving it for publication was 
Prof.\ Tobias Oechtering. \textit{(Corresponding author: Francesco Verde)}.}
\thanks{
D.~Darsena is with the Department of Engineering,
Parthenope University, Naples I-80143, Italy (e-mail: darsena@uniparthenope.it).
G.~Gelli and F.~Verde are with the Department of Electrical Engineering and
Information Technology, University Federico II, Naples I-80125,
Italy [e-mail: (gelli,f.verde)@unina.it].
I.~Iudice is with Italian Aerospace Research Centre (CIRA),
Capua I-81043, Italy (e-mail: i.iudice@cira.it).}
\thanks{D.~Darsena, G.~Gelli, and F.~Verde are also with 
National Inter-University Consortium for Telecommunications (CNIT).}}

\markboth{IEEE TRANSACTIONS ON INFORMATION FORENSICS
AND SECURITY,~Vol.~xx, No.~yy,~Month~2020}{Design and performance analysis of channel estimators under a pilot spoofing attack
in multiple-antenna systems}

% IEEE copyright notice added
\IEEEpubid{\begin{minipage}{\textwidth}\ \\
\center 1556--6013~\copyright~2020 IEEE. \\
Personal use is permitted, but republication/redistribution requires IEEE permission. \\
See http://www.ieee.org/publications\_standards/publications/rights/index.html for more information.
\end{minipage}}

\maketitle
\acceptednotice

\begin{abstract}

In multiple antenna systems employing time-division duplexing,
spatial precoder design at the base station (BS) leverages
channel state information acquired through uplink pilot transmission,
under the assumption of channel reciprocity.
Malicious eavesdroppers can start pilot spoofing
attacks to alter such design,
in order to improve their eavesdropping
performance in downlink.
The aim of this paper is to study
the effects of pilot spoofing attacks on
uplink channel estimation, by assuming that the
BS knows the angle
of arrivals (AoAs) of the
legitimate channels.
Specifically, after assessing the performance of
the simple least squares estimator (LSE),
we consider more sophisticated estimators, such as
the maximum likelihood estimator (MLE)
and different versions of the minimum mean
square error estimator (MMSEE), involving different degrees of
\textit{a priori} information about the
pilot spoofing attacks.
Theoretical analysis and numerical simulations are used to
compare the performance of such estimators.
In particular, we analy\-tically demonstrate that
the spoofing effects in the high signal-to-noise ratio regime
can be completely suppressed, under certain
conditions involving the AoAs of the legitimate
and spoofing channels.
Moreover,  we show that even an imperfect knowledge of the
AoAs and of the average transmission power of the spoofing signals
allows the MLE and MMSEE to achieve
significant performance gains over the LSE.

\end{abstract}


\section{Introduction}

\IEEEPARstart{M}{ultiple} transmit/receive antennas is a
well-established technology
to design high-speed reliable
wireless communication links \cite{Telatar}.
In terms  of  spectral  efficiency,
a  multiple-antenna system can
approach  the Shannon capacity of the
wireless channel, provided that
channel state information (CSI) is available at both
ends of the communication link \cite{Foschini-1,Foschini-2}.
In particular, design of effective transmit beamformers
for the downlink channel requires
knowledge of the CSI at the base station (BS).
A widely-adopted approach for acquiring CSI at the BS
in time-division duplexing (TDD) systems
relies on channel reciprocity.
Standard TDD beamforming protocols involve two phases:
in the {\em training phase}, the users
transmit known (\textit{pilot} or \textit{training}) symbols
to the BS in the uplink;
in the {\em data phase},
the BS estimates the uplink channels
and, capitalizing on channel reciprocity, forms the
downlink data beams to the users,
by using the estimated uplink channels.
TDD beamforming is expected to play a significant
role in forthcoming 5G systems \cite{Darsena.Wiley}, especially when working at
millimeter-wave frequencies and with
a huge number of antenna elements
at the BS (\textit{massive MIMO} systems).

\begin{figure*}
\centering
\includegraphics[width=1.25\columnwidth]{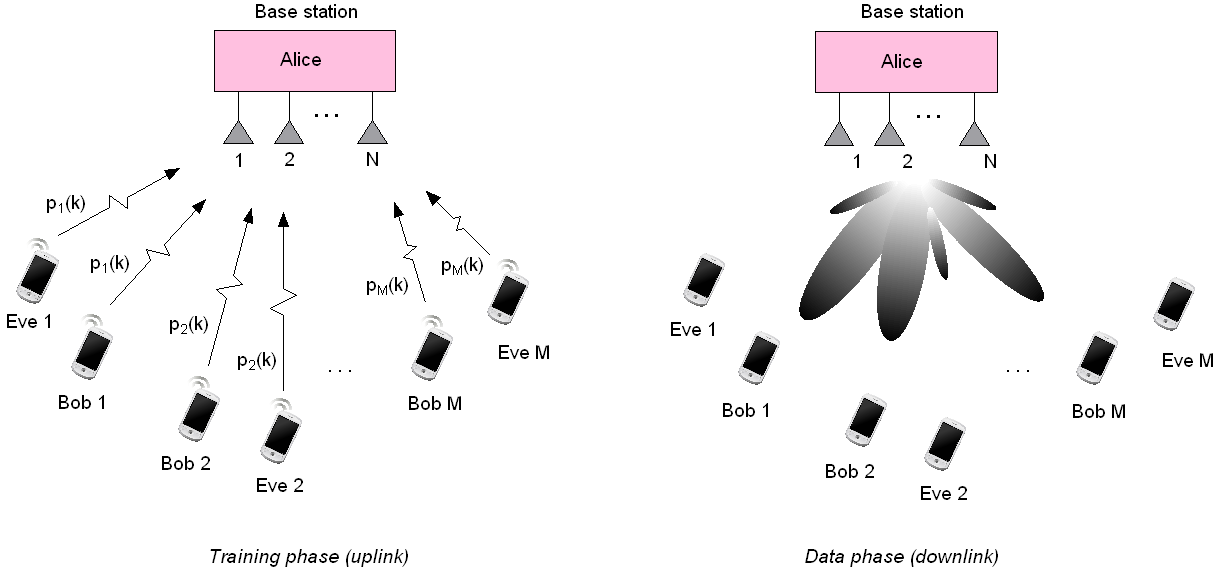}
\caption{Single-antenna eavesdroppers (Eves) launch
pilot spoofing attacks during the uplink training phase
between legitimate single-antenna users (Bobs) and
a multi-antenna BS (Alice), with the aim of
modifying beam design in order to intercept information
in the subsequent downlink data phase.}
\label{fig:fig_1}
\end{figure*}

\IEEEpubidadjcol

Uplink training sessions in TDD systems are vulnerable to malicious
attacks aimed at altering their operation \cite{Miller}.
Specifi\-cally, \emph{intelligent} eavesdroppers
can significantly enhance their wiretapping capability
in downlink by mounting
\textit{pilot spoofing attacks} in uplink \cite{Zhou}.
With reference to Fig.~\ref{fig:fig_1}, we consider a scenario
encompassing a multi-antenna BS, referred to as {\em Alice},
multiple legitimate single-antenna users, referred to as {\em Bobs}, and
multiple  single-antenna eavesdroppers, referred to as {\em Eves}.
The eavesdroppers in Fig.~\ref{fig:fig_1}
are intelligent in the sense that they can
act both as passive nodes (i.e., they try to intercept confidential communications
in downlink) as well as active ones (i.e., they can
transmit pilot signals in uplink).
Specifically,
the aim of each Eve is to steal information from a particular Alice-to-Bob
downlink transmission during the data phase, thus acting
as a passive terminal. With this goal in mind,
each Eve mounts a spoofing attack
during the training phase,
by transmitting the same pilot sequence as the selected Bob,
thus operating as an active node.
In this paper, we focus on such pilot attacks during the
uplink training phase and study the
performance of the channel estimation
process at Alice, by assuming that Bobs
employ orthogonal pilot sequences.

Several papers have dealt with the problem of combating pilot spoofing attacks
\cite{Kap,Im,Xiong.2015,Bas,Tugnait.2015,Xiong.2016,Wu,Tugnait.2017,Tugnait.2018,Huang,Wang.2019}. In all such papers, a common belief is that, if Eves attack the uplink
training phase by transmitting the same training sequence as Bobs, the channel
estimation process at Alice is irreparably compromised. Henceforth, all the efforts
in \cite{Kap,Im,Xiong.2015,Bas,Tugnait.2015,Xiong.2016,Wu,Tugnait.2017,Tugnait.2018,Huang,Wang.2019} have been mainly directed towards detection of the Eves' attacks
and relative countermeasures.
In particular, an energy-ratio-based detector has been proposed in
\cite{Xiong.2015} to detect a spoofing attack, which
explores the asymmetry of the received signal power levels
between Bob and Alice.
An attack can be additionally detected by resorting to
the source enumeration approach \cite{Tugnait.2015} and
the minimum description length criterion \cite{Tugnait.2017}.
To mitigate the effects of spoofing, estimation of both
Bob and Eve channels has been studied in \cite{Tugnait.2018},
along with secure beamforming.
A random channel training scheme has been proposed in \cite{Wang.2018},
where multiple orthogonal pilot sequences
are simultaneously allocated to Bob,
who randomly selects
one pilot sequence to transmit.
A remarkable extension of \cite{Tugnait.2018} and \cite{Wang.2018} to a multiuser
scenario has been developed in \cite{Wang.2019}.

Although all the aforementioned works have collectively provided a
significant research progress on counteracting pilot spoofing attacks,
they did not study in detail the effects of spoofing
on channel estimation performance.
Moreover, in such papers, the adopted fading channel models do not
incorporate detailed spatial information,
such as the angle-of-arrivals (AoAs), which
will be shown in the following to play a central role when
assessing the legitimate channel estimation
performance under a spoofing attack.

In this paper, we resort to angle-dependent multipath fading channel
models for the Bobs-to-Alice and Eves-to-Alice  links,
by assuming that the
AoAs of the legitimate
Bobs-to-Alice channels are known at Alice.
Moreover, we consider
\textit{least squares (LS)}, \textit{maximum likelihood (ML)}, and
\textit{minimum mean square error (MMSE)} channel estimators \cite{Kay}.

LS and ML estimators belong to the family of
the classical estimation approaches, for which
the legitimate channels are viewed as
\textit{deterministic but unknown} vectors.
In the case of the LS estimator (LSE),
which is the easiest type of estimator to
analyze and understand, the existence of the spoofing
signals is ignored and, thus, no information
about the pilot spoofing attacks is required at Alice.
On the other hand, when the spoofing signals are treated as
purely noise for the derivation of the ML estimator (MLE),
knowledge of the correlation matrix of each spoofing-plus-noise
contribution is also required.

Following the \textit{Bayesian approach}, three MMSE estimators (MMSEEs) are
further studied in this paper by
regarding the legitimate channels  as  \textit{random}
vectors whose particular realizations have to be estimated.
For each Bob-to-Alice channel, the (optimal) MMSE estimator (MMSEE)  minimizes
the \textit{Bayesian mean squared error (BMSE)}, where
the average is taken not only over the
data -- which also includes the contribution of the
corresponding Eve -- but over the probability density function (pdf) of
the legitimate channel vector as well, without imposing any
constraint on the structure of the estimator.
The MMSEE
has prior information on the legitimate channels
and exploits this information to outperform the MLE,
especially for low-to-moderate signal-to-noise ratio (SNR) values.

The remaining two MMSEEs are developed
when Alice is supposed to be aware of
the pilot spoofing attacks: namely, Alice has
\textit{a priori} estimates  of the AoAs and average transmission power
of each spoofing signal, and she
knows the statistical
characterization of the corresponding estimation errors.
In this case, for each Bob-to-Alice channel,
the estimator minimizing the Bayesian MSE --
where the average is now taken over the joint pdf of the
legitimate channel, the data, and the estimation errors
of the corresponding Eve's parameters -- does not admit a closed-form expression.
To obtain mathematically tractable solutions, we
retain the Bayesian MMSE criterion but constrain the estimators to be linear.

With reference to a given Bob-Eve pair,
the contribution of this paper is threefold:

\begin{enumerate}[1)]

\item
We show analytically that
the LSE does not have spoofing suppression capabilities,
even in the absence of noise, whereas the MLE can perfectly
reject the spoofing pilots of Eve in the high-SNR region,
provided that the subspaces generated by the columns of the
steering matrices
of the  legitimate and spoofing channels are nonoverlapping.

\item
We analytically demonstrate that,
if the AoAs through the multipath
Bob-to-Alice and Eve-to-Alice
channels are all distinct, the MMSEE is able to cancel the spoofing signal
in the high-SNR regime.
The synthesis of the MLE and MMSEE involves the knowledge of the
correlation matrix of the received data. We discuss how in principle such a matrix
can be estimated and, hence, MLE and MMSEE can be implemented without
requiring awareness of the pilot spoofing attack.

\item
We develop two (suboptimal) linear MMSEEs (LMMSEEs)
estimators: the former one, which is referred to as \textit{naive} LMMSEE,
is synthesized by relying on estimates of the AoAs and average transmission
power of Eve, without making  use of the statistical characterization of the
corresponding estimation errors; the latter one,
which is referred to as \textit{improved} LMMSE,
besides using the
estimates of the Eve's parameters, also exploits the knowledge of the pdf
of the estimation errors.
We prove that the naive LMMSEE exhibits a serious performance
degradation in the high SNR regime, whereas its improved version is fairly
resistant to the pilot spoofing attack.

\end{enumerate}

The paper is organized as follows. The next section
describes the signal model and introduces some basic assumptions.
Section~\ref{sec:LSE&MLE} revisits LSE and MLE,
and reports a theoretical analysis of their spoofing suppression capabilities.
The performance of the MMSEE and its implementation from data
are studied in Section~\ref{sec:MMSEE}; naive and
improved  LMMSEEs are
developed in the same section starting from estimates of the Eve's parameters.
Section~\ref{sec:simulation} discusses analytical and simulation
results in terms of BMSE and secrecy rate, and
Section~\ref{sec:concl} offers some conclusions.

\subsection{Notations}
Upper- and lower-case bold letters denote matrices and vectors;
the superscripts
$*$, $\trasp$, $\herm$, $-1$, and $\dag$
denote the conjugate,
the transpose, the Hermitian (conjugate transpose),
the inverse,
and the Moore-Penrose generalized
inverse \cite{Ben.book.2002} of a matrix;
$\mathbb{C}$, $\mathbb{R}$ and $\mathbb{Z}$ are
the fields of complex, real and integer numbers;
$\mathbb{C}^{n}$ $[\mathbb{R}^{n}]$ denotes the
vector-space of all $n$-column vectors with complex
[real] coordinates;
similarly, $\mathbb{C}^{n \times m}$ $[\mathbb{R}^{n \times m}]$
denotes the vector-space of all the $n \times m$ matrices with
complex [real] elements;
$\text{erf}(\cdot)$ denotes the error function \cite{Proakis};
$\delta(n)$ is the Kronecker delta, i.e., $\delta(n)=1$ when $n=0$
and zero otherwise;
$\log_a(\cdot)$ is taken to the base $a$ [we also use the shorthand
$\ln(\cdot)$ for $\log_e(\cdot)$]
and $j \eqdef \sqrt{-1}$ denotes the imaginary unit;
$\mathbf{0}_{n}$, $\Zero_{n \times m}$ and $\Id_{n}$
denote the $n$-column zero vector, the $n \times m$
zero matrix and the $n \times n$ identity matrix;
$\otimes$ denotes the Kronecker product between two matrices \cite{Horn.book.1991};
for any $\mathbf{A} \in \mathbb{C}^{n \times m}$,
$\mathrm{rank}(\mathbf{A})$ and $\mathrm{trace}(\mathbf{A})$
denote the rank and the trace of $\mathbf{A}$;
for any $\mathbf{A} \in \mathbb{C}^{n \times m}$,
$\|\bA\| \eqdef [\trace(\bA\,\bA^\herm)]^{1/2}$ denotes the (induced)
Frobenius norm of $\bA$ \cite{Horn.book.1990};
matrix $\mathbf{A}= \diag (a_{0}, a_{1}, \ldots,
a_{p-1}) \in \mathbb{C}^{p \times p}$
is diagonal;
$\frac{\partial}{\partial \bA^*} f \in \mathbb{C}^{n \times m}$ is the
{\em gradient} of the real-valued
scalar function $f$ with respect to $\bA^* \in \mathbb{C}^{n \times m}$ \cite{Hj};
the operator $\vec(\bA)\eqdef [\ba_1^\trasp,
\ba_2^\trasp, \ldots, \ba_m^\trasp]^\trasp \in \mathbb{C}^{n m}$ creates a column vector from
the matrix $\mathbf{A}=[\ba_1, \ba_2, \ldots, \ba_m] \in \mathbb{C}^{n \times m}$
by stacking the column vectors of $\bA$;
$\nullo(\A)$,
$\range(\A)$, and
$\Orange(\A)$
denote the null space, the
range (column space), and the
orthogonal complement of the column
space  of $\A \in \C^{n \times m}$ in
$\C^n$;
$\Es[\cdot]$ denotes ensemble averaging
and, finally, a circularly symmetric complex Gaussian random vector
$\bx \in \C^n$ with mean $\bm{\mu} \in \C^n$ and covariance
matrix $\bR \in \C^{n \times n}$ is denoted as
$\bx \sim {\cal CN}(\bm{\mu},\bR)$.

\section{System model and preliminaries}
\label{sec:model}

As shown in Fig.~\ref{fig:fig_1},
our scenario encompasses one BS (Alice) equipped with $\Nr$
receive antennas, $M$ legitimate single-antenna users (Bobs) and $M$ single-antenna
eavesdroppers (Eves).
Bobs use orthogonal pilot sequences for channel estimation:
specifi\-cally, for $\mathcal{K} \eqdef \{0,1,\ldots,K-1\}$
and $\mathcal{M} \eqdef \{1,2,\ldots,M\}$,
let $p_m(k) \in \C$ denote the pilot symbol transmitted
by the $m$th Bob within
the $k$th symbol interval, with
average transmission power
(per symbol) $\Pam$ and $K \ge 1$ \cite{Biguesh.2006}.%
\footnote{Throughout the paper,
we use the convention that the subscripts B and E
indicate a quantity referring to Bobs and Eves, respectively.}
Orthogonality among the pilot
sequences means that
$\bp_{m_1}^\herm \, \bp_{m_2}=\delta(m_1-m_2)$
where $\bp_m \eqdef [p_m(0), p_m(1), \ldots, p_m(K-1)]^\trasp \in \C^K$ is the
pilot vector of the $m$th Bob.
The vector $\bp_m$ is known to the $m$th Eve, which concurrently sends the same pilot
block in the training phase with
average transmission power $\Psm$.
Moreover, $\bp_1, \bp_2, \ldots, \bp_M$ are perfectly known at Alice.

We assume that no appreciable local scattering occurs at Alice
and, hence, fading at its antennas
is spatially correlated.
Such an assumption is reasonable \cite{Ertel} when the BS
is sufficiently high above the ground.
In this case, for narrowband signals,
the angle-dependent multipath
fading single-input multiple-output (SIMO)
baseband channel between
the generic transmitter $\text{TX} \in \{\text{B}, \text{E}\}$
and Alice can be modeled as
\be
\widetilde{\bh}_{\text{TX},m} = \frac{1}{\sqrt{L_{\text{TX},m}}}
\sum_{\ell=1}^{L_{\text{TX},m}}
\ba(\theta_{\text{TX},\ell,m}) \,  h_{\text{TX},\ell,m}
\label{eq:channel}
\ee
where $L_{\text{TX},m} \in \N$ is the number of 
paths between the $m$th transmitter
and Alice,
$\ba(\theta_{\text{TX},\ell,m}) \in \C^N$ and $h_{\text{TX},\ell,m} \in \C$
denote the steering vector and the gain of the $\ell$th 
path of the $m$th Bob, respectively,
with $\theta_{\text{TX},\ell,m}$ being the corresponding AoA.
For the sake of
simplicity, we assume a uniform linear array (ULA) at Alice.\footnote{Our framework can be extended to nonuniform arrays as well.}
The (normalized) steering vector $\ba(\theta_{\text{TX},\ell,m})$ can be expressed \cite{Fusco} as shown at the top of the next page in \eqref{eq:a},
\begin{figure*}[!t]
\normalsize\be
\ba(\theta_{\text{TX},\ell,m}) = \frac{\left[e^{j 2 \pi \Psi_0(\theta_{\text{TX},\ell,m})},
e^{j 2 \pi \Psi_1(\theta_{\text{TX},\ell,m})}, \ldots,
e^{j 2 \pi \Psi_{N-1}(\theta_{\text{TX},\ell,m})} \right]^\trasp}{\sqrt{N}}
\label{eq:a}
\ee
\hrulefill
\end{figure*}
with
\be
\Psi_n(\theta_{\text{TX},\ell,m}) \eqdef  - n \, \frac{d}{\lambda_\text{c}} \cos(\theta_{\text{TX},\ell,m})
\label{eq:Psi}
\ee
 for $n \in \{0,1,\ldots, N-1\}$, where $d$ is the absolute
antenna spacing and $\lambda_\text{c}$ is the signal wavelength.
We assume that legitimate and spoofing signals
are perfectly synchronized \cite{Miller}.

The discrete-time signal vector $\by(k) \in \C^{\Nr}$
received by Alice within the $k$th symbol period can be expressed as
\begin{multline}
\by(k) = \sum_{m=1}^M \left(\sqrt{\Pam} \, \bA_{\text{B},m} \, \bhbm
\right. \\ \left.
+ \sqrt{\Psm}  \, \bA_{\text{E},m} \,  \bh_{\text{E},m} \right) p_m(k) + \bv(k)
\label{eq:sig-spatial}
\end{multline}
with $k \in \mathcal{K}$, where, for $\text{TX} \in \{\text{B}, \text{E}\}$,
\begin{multline}
\bA_{\text{TX},m}  \eqdef \frac{1}{\sqrt{L_{\text{TX},m}}}
\left[ \ba(\theta_{\text{TX},1,m}), \ba(\theta_{\text{TX},2,m}), \ldots,
\right. \\ \left.
\ba(\theta_{\text{TX},L_{\text{TX},m},m})\right] \in \C^{N \times L_{\text{TX},m}}
\label{eq:A}
\end{multline}
\be
\bh_{\text{TX},m}  \eqdef \left[ h_{\text{TX},1,m}, h_{\text{TX},2,m}, \ldots,
h_{\text{TX},L_{\text{TX},m},m} \right]^\trasp \in \C^{L_{\text{TX},m}}
\label{eq:h}
\ee
and
$\bv(k) \sim {\cal CN}(\mathbf{0}_{\Nr},\sigma_v^2 \, \bI_{\Nr})$
is additive white Gaussian noise,
with $\bv(k_1)$ and $\bv(k_2)$ statistically
independent of each other for $k_1 \neq k_2 \in \mathcal{K}$.
Hereinafter, $\bA_{\text{TX},m}$ is referred to as the $m$th \textit{steering matrix}.
We assume a Rayleigh fading model, according to which
$\bhbm \sim
{\cal CN}(\mathbf{0}_{L_{\text{B},m}}, \bI_{L_{\text{B},m}})$ and
$\bh_{\text{E},m}
\sim  {\cal CN}(\mathbf{0}_{L_{\text{E},m}}, \bI_{L_{\text{E},m}})$
are mutually independent vectors, statistically independent of
$\bv(k)$, $\forall k \in \mathcal{K}$ and $ \forall m \in \mathcal{M}$.

Let us gather all the data \eqref{eq:sig-spatial}
received during the uplink pilot phase in
$\bY \eqdef [\by(0), \by(1), \ldots, \by(K-1)] \in \C^{\Nr \times K}$, thus obtaining
the signal model\footnote{For $\text{TX} \in \{\text{B}, \text{E}\}$,
$\EuScript{P}_{\text{TX},m}$ also represents the average received power
during the training phase, since
$\Es [\|\bA_{\text{TX},m} \, \bh_{\text{TX},m} \, \bp^\trasp\|^2 ]=1$.
}
\begin{multline}
\bY=\sum_{m=1}^M \left(\sqrt{\Pam} \, \bA_{\text{B},m} \, \bh_{\text{B},m}
\right. \\ \left. +  \sqrt{\Psm}  \, \bA_{\text{E},m} \,  \bh_{\text{E},m} \right) \bp_m^\trasp + \bV
\label{eq:sig-mat}
\end{multline}
where $\bV \eqdef [\bv(0), \bv(1), \ldots, \bv(K-1)] \in \C^{\Nr \times K}$.

Let us focus on the channel estimation process of the $\overline{m}$th user,
with $\overline{m} \in \mathcal{M}$. In this case, by capitalizing on
the orthogonality among the pilot vectors,  Alice performs the
correlation of the received data $\bY$ with $\bp_{\overline{m}}$, thus
obtaining
\be
\by_{\overline{m}} \eqdef \bY \, \bp_{\overline{m}}^* = \bK_{\text{B},\overline{m}}
\, \bh_{\text{B},\overline{m}} +
\bK_{\text{E},\overline{m}} \, \bh_{\text{E},\overline{m}} + \bv_{\overline{m}}
\label{eq:vec-model}
\ee
where we have defined $\bK_{\text{B},\overline{m}} \eqdef
\sqrt{\Pamoverline} \, \bA_{\text{B},\overline{m}} \in
\C^{\Nr \times L_{\text{B},\overline{m}}}$ and
$\bK_{\text{E},\overline{m}} \eqdef
\sqrt{\Psmoverline} \bA_{\text{E},\overline{m}} \in
\C^{\Nr \times L_{\text{E},\overline{m}}}$,  and, by assumption,
$\bv_{\overline{m}} \eqdef \bV \, \bp_{\overline{m}}^* \sim
{\cal CN}(\mathbf{0}_{\Nr},\sigma_v^2 \, \bI_{\Nr})$.

To simplify the notation, in the remaining part of the paper,
we will drop the subscript $\overline{m}$
in \eqref{eq:channel}--\eqref{eq:Psi}, \eqref{eq:A}--\eqref{eq:h},
and \eqref{eq:vec-model}, and study
different estimation strategies
for reliably acquiring the CSI of a generic Bob-to-Alice uplink in order to
design a suitable beamformer for the subsequent
Alice-to-Bobs downlink data transmission.
In all the considered cases, we assume that Alice has perfect knowledge of
the composite matrix $\bK_\text{B}$ (depending on
the transmit power $\Pa$ and
the steering matrix $\bA_\text{B}$). While  $\Pa$ is a system
parameter that is known \textit{a priori}, the matrix $\bA_\text{B}$
has to be estimated by Alice.
However, such a steering matrix varies much slower than $\bhb$ and, thus, it
can be estimated \cite{Tuncer,Chen,Chandran} in practice during a secure setup session.
Consequently, the matter boils down to estimate $\bhb$.\footnote{In principle,
Alice may perform a channel-unaware beamforming in downlink by relying only
on the knowledge of $\bK_\text{B}$, so-called \textit{angular beamforming}, at the
price of a capacity degradation \cite{AB}. Such a degradation might be even more
severe in terms of secrecy capacity in the presence of malicious Eves.}

In this paper, we study two different estimation strategies:
in the former one, following the classical approach, the entries of $\bhb$ are
assumed to be deterministic but unknown constants;
in the latter one, according to the Bayesian philosophy,
the knowledge of the pdf
of $\bhb$ is exploited to estimate its  parti\-cular realization.
As a performance measure of the considered estimators, we resort to
the BMSE, which is defined as
\be
\text{BMSE}(\widehat{\bh}_{\text{B}}) \eqdef \Es \left[\left \|
\widehat{\bh}_{\text{B}} -\bh_\text{B} \right\|^2 \right]
\label{eq:AMSE}
\ee
where $\widehat{\bh}_{\text{B}} \in \C^{L_\text{B}}$ denotes an estimate of
$\bhb$ and, unless otherwise specified, the expectation is taken with respect to the
pdf of the triple
$(\bhb, \bh_{\text{E}}, \bv)$.\footnote{
\label{foot:5}The
BMSE is a reasonable performance metric in fading channels not only
for Bayesian estimators, but also for the classical ones
that are designed under the deterministic assumption for $\bhb$.}

\vspace{-3mm}

\section{Least squares and maximum likelihood estimators}
\label{sec:LSE&MLE}

The channel estimators derived  herein  are
based on the assumption that both $\bhb$
and $\bK_\text{E}$ are deterministic
but unknown quantities. In this case, we can develop different
\textit{unbiased} estimators on the basis of the amount of knowledge regarding the
spoofing attack. When Alice is unaware of
the Eve's presence, the LSE can be used by Alice to estimate
$\bhb$. On the other hand, if Alice has perfect knowledge of the
correlation matrix of the \textit{disturbance (i.e., spoofing signal plus noise)}
(see the forthcoming discussion), it can implement the MLE to accomplish the same task.

\subsection{Least squares estimator}
\label{sec:LSE}

The LSE is defined \cite{Kay} as
\be
\widehat{\bh}_{\text{B},\text{LS}} \eqdef \arg \min_{\bh_\text{B} \in \C^{L_\text{B}}}
\left\|\by - \bK_\text{B} \, \bhb\right\|^2 \: .
\label{eq:LSE-def}
\ee
Under the assumption
that $\bK_\text{B}$ is full-column rank, i.e.,
$\rank(\bK_\text{B})=\rank(\bA_\text{B})=L_\text{B}$, the LSE is {\em unique} and it can
be written as (see \cite{Kay})
\be
\widehat{\bh}_{\text{B},\text{LS}} = \left(\bK_\text{B}^\herm  \, \bK_\text{B}\right)^{-1} \bK_\text{B}^\herm   \, \by \: .
\label{eq:LSE}
\ee
Since $\bA_\text{B}$ is a Vandermonde-like matrix, the condition
$\rank(\bK_\text{B})=\rank(\bA_\text{B})=L_\text{B}$ is fulfilled \cite{Horn.book.1990} if
$N \ge L_\text{B}$
and
$\theta_{\text{B},1} \neq \theta_{\text{B},2} \neq \cdots \neq
\theta_{\text{B},L_\text{B}}$ (i.e., the AoAs through the multipath
channel between Bob and Alice are distinct).
On the other hand,  if $\rank(\bK_\text{B}) < L_\text{B}$, problem
\eqref{eq:LSE-def} has infinitely many solutions and
$\bhb$ is not identifiable, i.e., if $\widehat{\bh}_{\text{B},\text{LS}}$
is a solution of $\eqref{eq:LSE-def}$ and $\balpha \in \nullo(\bK_\text{B})$, then
$\widehat{\bh}_{\text{B},\text{LS}}+\balpha$
is another solution of \eqref{eq:LSE-def}.
The synthesis of the LSE  involves knowledge of $\bK_\text{B}$ only
and its computational burden is dominated by the matrix inversion
in \eqref{eq:LSE}, which involves
$\mathcal{O}(L_\text{B}^3)$ floating point
operations (flops) \cite{Loan} if computed from scratch.

The LSE \eqref{eq:LSE} is unbiased
and, by using the properties of the
Kronecker product and the trace operator,
its BMSE \eqref{eq:AMSE}  can
be expressed as follows
\begin{multline}
\text{BMSE}_\text{LS}
\eqdef \frac{\trace\left[\bA_\text{B} \left(\bA_\text{B}^\herm \, \bA_\text{B}\right)^{-2}
\bA_\text{B}^\herm \, \bA_\text{E} \, \bA_\text{E}^\herm \right]}{\text{SSR}}
\\ +
\frac{\trace\left[\left(\bA_\text{B}^\herm \, \bA_\text{B}\right)^{-1}\right]}{\text{SNR}_\text{B}}
\label{eq:AMSE-LSE}
\end{multline}
where we have defined the SNR of the Bob transmission
as $\text{SNR}_\text{B} \eqdef \Pa /\sigma_v^2$, whereas
$\text{SSR} \eqdef \Pa/{\Ps}$ represents the
signal-to-spoofing ratio (SSR).
It is noteworthy from \eqref{eq:AMSE-LSE}
that, in the absence of noise, i.e.,
as $\sigma_v^2$ approaches to zero, the BMSE of the LSE
exhibits a saturation effect, namely, a floor given by
\barr
\overline{\text{BMSE}}_\text{LS} & \eqdef \lim_{\sigma_v^2 \to 0}
\text{BMSE}_\text{LS}
\nonumber \\ &
=  \frac{\trace\left[\bA_\text{B} \left(\bA_\text{B}^\herm \, \bA_\text{B}\right)^{-2}
\bA_\text{B}^\herm \, \bA_\text{E} \, \bA_\text{E}^\herm \right]}{\text{SSR}}
\label{eq:AMSE-LSE-asympt}
\earr
which is due to the malicious pilot transmission of Eve.
The following lemma provides bounds on
$\overline{\text{BMSE}}_\text{LS}$.

\begin{lemma}
\label{lem:1}
The BMSE floor of \eqref{eq:LSE} is bounded  as
\barr
\frac{1}{\text{SSR}} \sum_{\ell=1}^{L_\text{B}}
\frac{\sigma_{N-\ell+1}^2(\bA_\text{E})}{\sigma_{\ell}^2(\bA_\text{B})}
 \le \overline{\text{BMSE}}_\text{LS} & \le
\frac{1}{\text{SSR}} \sum_{\ell=1}^{L_\text{B}}
\frac{\sigma_{\ell}^2(\bA_\text{E})}{\sigma_{\ell}^2(\bA_\text{B})}
\label{eq:AMSE-LSE-asympt-bound}
\earr
where
$\sigma_{1}(\bA_\text{B}) \le \sigma_{2}(\bA_\text{B}) \le \cdots \le
\sigma_{L_\text{B}}(\bA_\text{B})$
are the \textit{nonzero} singular values
of $\bA_\text{B}$
arranged in increasing order and
$\sigma_{1}(\bA_\text{E}) \ge \sigma_{2}(\bA_\text{E}) \ge \cdots \ge
\sigma_{N}(\bA_\text{E})$
are the singular values of $\bA_\text{E}$
arranged in decreasing order.

\end{lemma}
\begin{proof}
See Appendix~\ref{app:lem-1}.
\end{proof}

Lemma~\ref{lem:1} enlightens that the spoofing attack
might seriously affect the performance of \eqref{eq:LSE}, which
depends not only on the $\text{SSR}$, but also on the ratio
between the singular values of $\bA_\text{E}$
and $\bA_\text{B}$. In particular,
the lower bound in \eqref{eq:AMSE-LSE-asympt-bound}
shows that, except for the limit case $\text{SSR} \to + \infty$,
the floor $\overline{\text{BMSE}}_\text{LS}$ cannot be
zero if the $L_\text{B}$ smallest singular values of
$\bA_\text{E}$ are nonzero. This happens when
the rank of $\bA_\text{E}$ is greater than the
number of paths between Bob and Alice, that is,
compared to the legitimate channel,
the spoofing one is characterized by a
richer scattering with significant multipath components.

\subsection{Maximum likelihood estimator}
\label{sec:ML}

Let $p(\by; \bhb)$ denote the pdf
of $\by$, parameterized by $\bhb$. The MLE is the solution of
the maximization problem
\be
\widehat{\bh}_{\text{B},\text{ML}}
\eqdef \arg \max_{\bhb \in \C^{L_\text{B}}} p(\by; \bhb) \: .
\label{eq:MLE}
\ee
Since the disturbance
$\d \eqdef \bK_\text{E} \, \bh_{\text{E}} + \bv \sim {\cal CN}(\mathbf{0}_{\Nr}, \bR_{\bd\bd})$, with $\bR_{\bd\bd} \eqdef \Es[\bd \, \bd^\herm]=
\bK_\text{E} \, \bK_\text{E}^\herm + \sigma_v^2 \, \bI_{\Nr}$
being its correlation matrix (depending on $\bK_\text{E}$),
it results that $\widehat{\bh}_{\text{B},\text{ML}}$ is the
solution of the matrix equation (see, e.g., \cite{Kay})
\be
\left(\bK_\text{B}^\herm \, \bR_{\bd\bd}^{-1} \, \bK_\text{B} \right)
\widehat{\bh}_{\text{B},\text{ML}} =
\bK_\text{B}^\herm \, \bR_{\bd\bd}^{-1} \, \by \: .
\label{eq:MLE-eq}
\ee
If $\bK_\text{B}$ is full-column rank, the MLE is unique
and given by
\be
\widehat{\bh}_{\text{B},\text{ML}} =
\left(\bK_\text{B}^\herm \, \bR_{\bd\bd}^{-1} \, \bK_\text{B} \right)^{-1}
\bK_\text{B}^\herm \, \bR_{\bd\bd}^{-1} \, \by \: .
\label{eq:MLE-sol}
\ee
On the other hand, if the columns of $\bK_\text{B}$
are linearly dependent, there exists an infinite number of
solutions for  \eqref{eq:MLE-eq}
that generate the same density function and, thus,
similarly to the LSE, $\bhb$ is not identifiable in this case.

Compared to the LSE in \eqref{eq:LSE}, the MLE requires
the additional knowledge of the correlation matrix of the
disturbance, which in its turn depends on the noise variance
$\sigma_v^2$ and
the composite matrix $\bK_\text{E}$ (determined by
the transmit power $\Ps$ and
the steering matrix $\bA_\text{E}$).
The noise variance $\sigma_v^2$ is related to the noise figure of
Alice and, thus, it can be known \textit{a priori} or estimated previously.
On the other hand, both
$\Ps$ and $\bA_\text{E}$ are unknown at Alice.
In principle, one can estimate $\bR_{\bd\bd}$ from the received data,
by observing that the correlation matrix of $\by$ can be expressed as
\be
\bR_{\by\by} \eqdef \Es[\by \, \by^\herm]=
\bK_\text{B} \, \bK_\text{B}^\herm + \bR_{\bd\bd} \: .
\label{eq:Ryy}
\ee
Given a \textit{sample} estimate $\bS_{\by\by}$ of
$\bR_{\by\by}$ and knowledge of $\bK_\text{B}$, a corresponding sample estimate of
$\bR_{\bd\bd}$ can be obtained as
$\bS_{\bd\bd}=\bS_{\by\by}-\bK_\text{B} \, \bK_\text{B}^\herm$.
Estimation
of $\bR_{\by\by}$ from the received data will be discussed in
Subsection~\ref{sec:MMSE-1}.
Since $\Nr$ is tipically much larger than $L_\text{B}$, the computational
complexity of the MLE is mainly dictated by the inversion of
$\bR_{\bd\bd}$, which requires $\mathcal{O}(\Nr^3)$
flops
if one resorts to batch algorithms.

The MLE \eqref{eq:MLE-sol} is unbiased and
its BMSE \eqref{eq:AMSE}  can be expressed \cite{Kay} as
\be
\text{BMSE}_\text{ML} = \trace
\left[ \left(\bK_\text{B}^\herm \, \bR_{\bd\bd}^{-1}
\, \bK_\text{B} \right)^{-1} \right] \: .
\label{eq:AMSE-ML}
\ee
The MLE is also an efficient estimator since it attains the
{\em standard (i.e., for nonrandom parameter estimation)}
Cramer-Rao lower bound (CRLB) \cite{Kay}, that is
$\text{BMSE}_\text{ML} \le \text{BMSE}(\widehat{\bh}_{\text{B}})$,
for any unbiased estimator $\widehat{\bh}_{\text{B}}$ of
$\bhb$. In particular, in the considered scenario, one gets
$\text{BMSE}_\text{ML} \le \text{BMSE}_\text{LS}$.

Similarly to the LSE, our aim is to characterize the spoofing suppression
capabilities of the MLE in the
high-SNR regime. Such a characterization is provided by
the following lemma.

\begin{lemma}
\label{lem:2}
If the subspaces $\range(\bA_\text{B})$ and $\range(\bA_\text{E})$
are \textit{nonoverlapping} or \textit{disjoint}, i.e.,
\be
\range(\bA_\text{B}) \cap \range(\bA_\text{E})=\{\mathbf{0}_{\Nr}\}
\label{eq:nooverlap}
\ee
then perfect spoofing cancellation is achieved in the absence of noise, that is
\be
\overline{\text{BMSE}}_\text{ML} \eqdef \lim_{\sigma_v^2 \to 0}
\text{BMSE}_\text{ML}=0 \: .
\ee
\end{lemma}
\begin{proof}
See Appendix~\ref{app:lem-2}.
\end{proof}

A consequence of the above lemma is that,
for \eqref{eq:nooverlap} to hold, it suffices that the columns of
$\bA_\text{B}$ and $\bA_\text{E}$ are
linearly independent so that the (augmented) matrix
$[\bA_\text{B}, \bA_\text{E}] \in \C^{N \times (L_\text{B}+L_\text{E})}$
has full column rank. This condition is fulfilled if
$N \ge L_\text{B}+L_\text{E}$ and
$\theta_{\text{B},1} \neq \cdots \neq
\theta_{\text{B},L_\text{B}} \neq
\theta_{\text{E},1} \neq \cdots \neq
\theta_{\text{E},L_\text{E}}$, i.e., the AoAs of the multipath
Bob-to-Alice and Eve-to-Alice
channels are all distinct.
In a nutshell, we can state that, compared to the LSE, the MLE can
effectively counteract the pilot spoofing attack, at the
price however of requiring the knowledge of the correlation matrix of the
spoofing-plus-noise signal.

When condition \eqref{eq:nooverlap} is
not satisfied, i.e., the AoA ranges of
Bob and Eve are overlapping, perfect spoofing cancellation is impossible, even in the absence of noise. However, overlapping between the subspaces $\range(\bA_\text{B})$ and
$\range(\bA_\text{E})$ depends on the distribution of the AoAs, which is governed by the physical propagation environment, as well as on the locations of Bob and Eve.
Therefore, it is unlikely that condition \eqref{eq:nooverlap}
is violated at all times, due to the random transmitter locations
and scattering effects. Moreover, numerical results in Section~\ref{sec:simulation} show that the MLE is robust when the difference between
the angles of incidence at which
Bob and Eve arrive at Alice tends to zero for a given path.

\section{Minimum Mean Square Error Estimators}
\label{sec:MMSEE}

In this section, we consider
the Bayesian approach to statistical estimation \cite{Kay},
by capitalizing on the fact that $\bhb \sim
{\cal CN}(\mathbf{0}_{L_\text{B}}, \bI_{L_\text{B}})$.
Such an approach is different from the classical one
pursued in Section~\ref{sec:LSE&MLE}.
Compared to LSE and MLE,
Bayesian estimators
improve estimation accuracy by exploiting
\textit{a priori}
information about the pdf of $\bhb$.
Moreover, the class of Bayesian
estimators is not restricted to the
unbiased ones \cite{VanTrees,Dong.2002}.

Herein, we consider three different
MMSEEs based on different \textit{a priori} information about the attack of Eve.
In the first one, Alice does
not have any information regarding
$\bK_\text{E}$, which is modeled
as a deterministic but unknown matrix
(as already done in Section~\ref{sec:LSE&MLE}).
In the second one, it is assumed
that Alice has an imperfect knowledge
$\widehat{\bK}_\text{E}$ of $\bK_\text{E}$.
In the third one, besides $\widehat{\bK}_\text{E}$,
Alice also knows the statistical
characterization of the corresponding error.

\subsection{Case 1: No \text{a priori} knowledge about $\bK_\text{E}$}
\label{sec:MMSE-1}

Under the assumption that $\bK_\text{E}$ is deterministic, the optimal MMSEE minimizes
\eqref{eq:AMSE} and is given \cite{Kay} by
the mean of the {\em posterior} distribution of
$\bhb$, i.e.,
\be
\bhbhatmmse
=\Es[\bhb \, | \, \by] \: .
\label{eq:generalBayes}
\ee
In this case, the vectors
$\by$ and $\bhb$ are jointly complex Gaussian
and, hence,
the conditional distribution of $\bhb \, | \, \by$ is complex
Gaussian, too. Therefore, the
MMSEE \eqref{eq:generalBayes} turns out to be \textit{linear} and
assumes the form (see, e.g., \cite{Kay})
\barr
\bhbhatmmseuno & =
\left(\bI_{L_\text{B}}+ \bK_\text{B}^\herm \, \bR_{\bd\bd}^{-1} \,  \bK_\text{B} \right)^{-1}
\bK_\text{B}^\herm \, \bR_{\bd\bd}^{-1} \, \by
\nonumber \\ &=
\bK_\text{B}^\herm \, \bR_{\by\by}^{-1} \, \by
\label{eq:MMSE-S1}
\earr
where the matrix inversion lemma \cite{Horn.book.1991} has been used.
The corresponding minimum BMSE is \cite{Kay}:
\be
\text{BMSE}_\text{MMSE} =\trace\left[\left(\bI_{L_\text{B}}+\bK_\text{B}^\herm \, \bR_{\bd\bd}^{-1} \,  \bK_\text{B}\right)^{-1}\right] \: .
\label{eq:BMMSE-1}
\ee
We recall that both the LSE and MLE require
that $\bK_\text{B}$ be full column rank. As discussed
in Subsection~\ref{sec:LSE}, such a condition is met if and only if
the number of antennas at Alice is not smaller than the
number of paths from Bob to Alice and, moreover,
the corresponding AoAs are distinct.
In contrast, it can be seen from \eqref{eq:MMSE-S1} that, in the case
at hand, the MMSEE requires
the invertibility of
$\bI_{L_\text{B}}+ \bK_\text{B}^\herm \, \bR_{\bd\bd}^{-1} \,  \bK_\text{B}$.
For this to hold, $\bK_\text{B}$ need not be full column rank.
Therefore, the MMSEE \eqref{eq:MMSE-S1} can exist even if $N < L_\text{B}$ and/or
$\theta_{\text{B},1}, \theta_{\text{B},2}, \ldots,  \theta_{\text{B},L_\text{B}}$
are not distinct. However, under these circumstances, as shown soon after, the performance of the MMSEE \eqref{eq:MMSE-S1} is adversely affected by the spoofing attack, even in the absence of noise.
The synthesis of \eqref{eq:MMSE-S1} requires $\mathcal{O}(\Nr^3)$
flops to invert $\bR_{\by\by}$.

It is shown in Appendix~\ref{app:FIM-1} that the MMSEE \eqref{eq:MMSE-S1} attains the \textit{Bayesian CRLB} \cite{VanTrees}. Therefore, since the error of the MMSEE cannot be larger than that of the
maximum {\em a posteriori probability estimator (MAPE)} \cite{VanTrees},
the MMSEE and MAPE are equal in this case.
Similarly to the MLE, the MMSEE \eqref{eq:MMSE-S1} can avoid the
spoofing attack in the high-SNR region, as stated by the following lemma.

\begin{lemma}
\label{lem:3}
If the matrix
$[\bA_\text{B}, \bA_\text{E}] \in \C^{N \times (L_\text{B}+L_\text{E})}$
has full column rank, then
\be
\overline{\text{BMSE}}_\text{MMSE}
\eqdef \lim_{\sigma_v^2 \to 0} \text{BMSE}_\text{MMSE}
=0 \: .
\ee
\end{lemma}
\begin{proof}
See Appendix~\ref{app:lem-3}.
\end{proof}

The full-column rank property of
$[\bA_\text{B}, \bA_\text{E}]$
is a  sufficient condition to ensure that the subspaces $\range(\bA_\text{B})$ and $\range(\bA_\text{E})$ are nonoverlapping (see Subsection~\ref{sec:ML}). Therefore,
both the MLE and MMSEE \eqref{eq:MMSE-S1}
perfectly reject the spoofing signal in the absence of noise,
provide that $\rank([\bA_\text{B}, \bA_\text{E}]) = L_\text{B}+L_\text{E}$.
As also pointed out at the end of Subsection~\ref{sec:ML}, violation of such a condition is
unlikely in practice. Moreover, the MMSEE is robust against a partial
overlap between the AoA ranges of Bob and Eve (see Section~\ref{sec:simulation}).

\begin{figure}
\centering
\includegraphics[width=1\columnwidth]{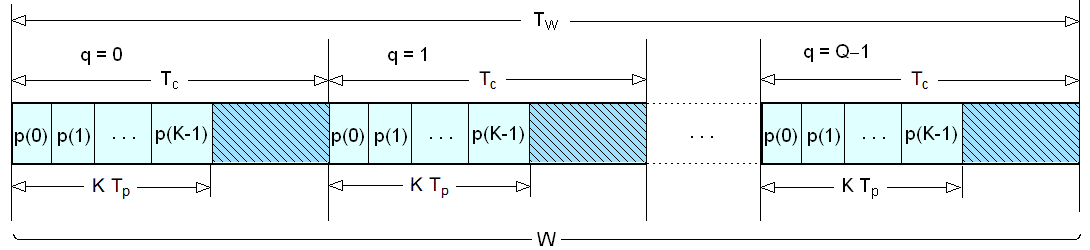}
\caption{A specific training design for estimating $\bR_{\by\by}$.}
\label{fig:fig_2}
\end{figure}

Apart from the knowledge of $\bK_\text{B}$,
the synthesis of the MMSEE \eqref{eq:MMSE-S1} requires estimation of
$\bR_{\by\by}$. Such a correlation matrix is the result
of an averaging operation taken over the noise, as well as
over the fading vectors $\bh_{\text{B}}$ and $\bh_{\text{E}}$.
Henceforth, estimation of $\bR_{\by\by}$ requires a dedicated
training session -- \textit{different from that used to estimate
$\bh_{\text{B}}$ through \eqref{eq:MMSE-S1}} --
spanning a time window $\mathcal{W}$, whose
duration $T_\text{w}$ is sufficiently larger
than the coherence time $T_\text{c}$ of the channel.
In principle, $\bR_{\by\by}$ can be consistently estimated by
the training scheme depicted in Fig.~\ref{fig:fig_2}.
Let $T_\text{p}$
be the period of the pilot symbols transmitted by Bob and Eve, the duration
$T_\text{w}$ can be divided in $Q$ coherence intervals of the channel,
i.e., $T_\text{w}= Q \, T_\text{c}$, and the number $K$ of pilot symbols $\bp$
(the same used to estimate $\bh_{\text{B}}$) can be
chosen such that $K \, T_\text{p}=T_\text{c}-T_\text{d}$, where
$T_\text{d}$ is the length of the downlink
information-bearing session.\footnote{We do not consider
uplink data transmissions in our discussion. However, if the legitimate users access
the uplink channel in an orthogonal fashion and Eves transmit jamming signals to
degrade the reception of the Bobs' data at Alice \cite{Wang.2019}, the estimation
accuracy of the correlation matrix $\bR_{\by\by}$ can be further improved by considering
uplink data symbols, too.}
If $\by^{(q)} \in \C^{\Nr}$ denote the data block \eqref{eq:vec-model} received by Alice during the
$q$th coherence interval of the channel within $\mathcal{W}$, for
$q \in \{0,1, \ldots, Q-1\}$, the correlation matrix $\bR_{\by\by}$
can be estimated as follows
\be
\bS_{\by\by} \eqdef \frac{1}{Q} \sum_{q=0}^{Q-1}
\by^{(q)} \left[\by^{(q)}\right]^\herm \: .
\label{eq:stima-Ryy}
\ee
Such a procedure gives a consistent estimate of $\bR_{\by\by}$,
provided that both Bob and Eve transmit during the time window $\mathcal{W}$
the same pilot vector $\bp$ used to estimate
$\bh_{\text{B}}$ through \eqref{eq:MMSE-S1}.
The
\textit{sample matrix inversion (SMI)} implementation of the MMSEE
\eqref{eq:MMSE-S1}
is obtained by replacing $\bR_{\by\by}$  in \eqref{eq:MMSE-S1}
with $\bS_{\by\by}$, that is
\barr
\bhbhatmmsesmi =
\bK_\text{B}^\herm \, \bS_{\by\by}^{-1} \, \by
\label{eq:MMSE-S1-smi}
\earr
which might exhibit a severe performance degradation
with respect to its ideal counterpart if
$Q$ is not sufficiently large.

To mitigate the performance degradation due to finite-
sample-size effects, one can resort to the \textit{subspace}
implementation of the MMSEE \eqref{eq:MMSE-S1},
by exploiting the properties of the eigenvalue decomposition
(EVD) of $\bR_{\by\by}$.\footnote{
Another viable alternative is represented by
shrinkage-based methods, which have the potential to enhance the performance of
correlation matrix estimation with small number of samples \cite{Stoica}.}
If
$[\bA_\text{B}, \bA_\text{E}]$
has full column rank (see Lemma~\ref{lem:3}),
accounting for \eqref{eq:Ryy} and
recalling that $\bR_{\bd\bd} =
\bK_\text{E} \, \bK_\text{E}^\herm + \sigma_v^2 \, \bI_{\Nr}$, the EVD
of $\bR_{\by\by}$ is given by
$\bR_{\by\by} =\bU_s \, \bLambda_s\,  \bU_s^\herm + \sigma_v^2 \, \bU_n \, \bU_n^\herm$, where $\bU_s \in \C^{\Nr \times (L_\text{B}+L_\text{E})}$ collects
the eigenvectors associated with the $L_\text{B}+L_\text{E}$ largest
eigenvalues $\lambda_1, \lambda_2, \ldots, \lambda_{L_\text{B}+L_\text{E}}$
of $\bR_{\by\by}$ (arranged
in decreasing order), whose columns span the \textit{signal subspace}
corresponding to the Bob and Eve transmissions,
i.e., the subspace
$\range(\bK)$ of $\bK \eqdef [\bK_\text{B}, \bK_\text{E}]
\in \C^{\Nr \times (L_\text{B}+L_\text{E})}$, while
$\bU_n \in \C^{\Nr \times (\Nr-L_\text{B}-L_\text{E})}$
collects the eigenvectors associated with the eigenvalue $\sigma_v^2$,
whose columns span the \textit{noise subspace}, i.e., the orthogonal
complement $\Orange(\bK)$ in $\C^{\Nr} $ of the subspace $\range(\bK)$
and, finally, $\bLambda_s \eqdef \diag(\lambda_1, \lambda_2, \ldots, \lambda_{L_\text{B}+L_\text{E}})$.
By substituting the EVD of $\bR_{\by\by}$ in
\eqref{eq:Ryy} and exploiting the orthogonality between signal and noise
subspaces, one equivalently obtains
\be
\bhbhatmmseuno  =
\bK_\text{B}^\herm \, \bU_s \, \bLambda_s^{-1}\,  \bU_s^\herm  \, \by \:.
\label{eq:MMSE-S1-sub-exact}
\ee
Since in practice the EVD is performed on $\bS_{\by\by}$ given by
\eqref{eq:stima-Ryy}, by denoting the sample matrices corresponding to
$\bU_s$ and $\bLambda_s$ with $\widehat{\bU}_s$ and $\widehat{\bLambda}_s$, respectively, one has
\be
\bhbhatmmsesub  =
\bK_\text{B}^\herm \, \widehat{\bU}_s \, \widehat{\bLambda}_s^{-1}\,
\widehat{\bU}_s^\herm  \, \by
\label{eq:MMSE-S1-sub} \:.
\ee
It is noteworthy that the estimator \eqref{eq:MMSE-S1-sub}
is not equal to \eqref{eq:MMSE-S1-smi}, since
$\bK_\text{B}^\herm \, \widehat{\bU}_n \neq \bO_{L_\text{B} \times (\Nr-L_\text{B}-L_\text{E})}$ due to the finite-sample  size effects,
where $\widehat{\bU}_n$ being the sample matrix of $\bU_n$.
This implies that
\eqref{eq:MMSE-S1-smi} and \eqref{eq:MMSE-S1-sub}
might exhibit different BMSE performances (see Subsection~\ref{sub_ex3}).
The estimator \eqref{eq:MMSE-S1-sub} basically demands the same computational burden as \eqref{eq:MMSE-S1-smi} and, additionally, requires the knowledge of the dimension
$L_\text{B}+L_\text{E}$ of the
signal subspace, which can be obtained from $\bS_{\by\by}$
by using the minimum description length criterion \cite{Wax}.

\subsection{Case 2: Imperfect knowledge of $\bK_\text{E}$}
\label{sec:MMSE-2}

Herein, our aim is to study the impact on the system performance of errors
regarding the knowledge of the transmission parameters of Eve.
Starting from an estimate of $\bR_{\bd\bd}$ or, equivalently, $\bR_{\by\by}$,
ML/MAP estimators
\cite{Book-VanTrees} or other computationally simpler
subspace-based estimation procedures \cite{Stoica-Music,Roy-Esprit} can be
used to estimate $\bK_\text{E}$.
Therefore, we assume that estimates of the AoAs and
the average transmit power of Eve are available at Alice, which are expressed
\cite{Tsai,Goldsmith} as\footnote{For the sake of analysis, we assume in Case 2
that the number $L_\text{E}$ of paths between Eve and Alice are also known.
In practice, an upper bound of $L_\text{E}$ might be available since, depending on the transmitted signal parameters (carrier frequency and bandwidth) and
application (indoor or outdoor), the maximum channel multipath
spread may be known \textit{a priori}.}
\barr
\widehat{\theta}_{\text{E},\ell} & = \theta_{\text{E},\ell} +
\Delta \theta_{\text{E},\ell} \: ,
\quad \text{for $\ell \in \{1,2, \ldots, L_\text{E}\}$}
\label{eq:dirhat}
\\
\Pshat & = \Ps \, e^{\Delta \Ps}
\label{eq:powhat}
\earr
where $\widehat{\theta}_{\text{E},\ell}$ and $\Pshat$
are estimates of
$\theta_{\text{E},\ell}$ and $\Ps$, respectively, whereas
the random variables (see Appendix~\ref{app:trunc})
\barr
\Delta \theta_{\text{E},\ell} & \sim {\cal N}_\text{T}
(0, \sigma_{\theta_{\text{E}}}, -\Delta \theta_{\text{E}, \text{max}}, \Delta \theta_{\text{E}, \text{max}})
\\
\Delta \Ps & \sim {\cal N}_\text{T}
(0, \sigma_{\Ps}, -\Delta\Psmax, \Delta\Psmax)
\earr
denote the
corresponding errors that model the \textit{uncertainty} on the knowledge
of the spoofing transmission parameters.
It is also assumed that
$\Delta\theta_{\text{E},1}, \Delta\theta_{\text{E},2}, \ldots,
\Delta\theta_{\text{E},L_\text{E}}$, and $\Delta\Ps$
are mutually independent random variables, statistically independent of
the triple $(\bhb, \bh_\text{E},\bv)$, whose probability distributions are known at Alice,
along with the noise variance $\sigma_v^2$ (see Subsection~\ref{sec:ML}
for a brief discussion about such an assumption).
As shown in Appendix~\ref{app:trunc}, the random variable
$e^{\Delta \Ps}$ in \eqref{eq:powhat} has a truncated lognormal distribution.

In this case, the optimal MMSEE minimizes
the BMSE \eqref{eq:AMSE},
where the average is taken not only over $(\bhb, \bh_\text{E},\bv)$, but
over the pdf of $(\Delta\theta_{\text{E},1}, \Delta\theta_{\text{E},2}, \ldots,
\Delta\theta_{\text{E},L_\text{E}},\Delta\Ps)$, too.
This estimator is {\em not} linear:
it is difficult to determine
in closed form and its computational complexity is prohibitive in practice.
Two mathematically tractable solutions are reported in the following
subsections by retaining the Bayesian MMSE criterion but constraining
the estimators to be {\em linear}.

\subsubsection{Naive LMMSEE}

As a first strategy to synthesize a LMMSEE
\cite{Kay} with affordable complexity, Alice can
disregard the knowledge of the statistics
of $\Delta\theta_{\text{E},1}, \Delta\theta_{\text{E},2}, \ldots, \Delta\theta_{\text{E},L_\text{E}},\Delta\Ps$ and simply build the estimate
$\widehat{\bK}_\text{E} \eqdef
\sqrt{\Pshat} \widehat{\bA}_\text{E}$
of $\bK_\text{E}$,
where $\widehat{\bA}_\text{E}$ is obtained from $\bA_\text{E}$
by replacing $\theta_{\text{E},\ell}$ with $\widehat{\theta}_{\text{E},\ell}$,
for $\ell \in \{1,2, \ldots, L_\text{E}\}$. So doing, an approximated
version of \eqref{eq:MMSE-S1} is developed as
\be
\bhbhatlmmseuno  =
\bK_\text{B}^\herm \, \widehat{\bR}_{\by\by}^{-1} \, \by
\label{eq:LMMSE-0}
\ee
with $\widehat{\bR}_{\by\by} \eqdef
\bK_\text{B} \, \bK_\text{B}^\herm  +
\widehat{\bK}_\text{E} \, \widehat{\bK}_\text{E}^\herm + \sigma_v^2 \, \bI_{\Nr}$.
In this case, the corresponding BMSE can be calculated by substituting
\eqref{eq:LMMSE-0} in \eqref{eq:AMSE} and,
by virtue of the conditional expectation rule \cite{Casella},
further averaging the obtained result
with respect to the pdf of
$(\Delta\theta_{\text{E},1}, \Delta\theta_{\text{E},2}, \ldots,
\Delta\theta_{\text{E},L_\text{E}},\Delta\Ps)$. So doing, one has
\begin{multline}
\text{BMSE}_\text{LMMSE}^{(1)} \eqdef
\Es \left\{
\trace\left[ \bK_\text{B}^\herm \,
\widehat{\bR}_{\by\by}^{-1}
\left(\bR_{\by\by} \, \widehat{\bR}_{\by\by}^{-1} - \bI_{\Nr} \right) \bK_\text{B} \right] \right\}
\\
+ \Es \left\{ \trace\left[\left(\bI_{L_\text{B}}+\bK_\text{B}^\herm \, \widehat{\bR}_{\bd\bd}^{-1} \,  \bK_\text{B}\right)^{-1}\right] \right \}
\label{BMSE-LMMSE-1}
\end{multline}
with $\widehat{\bR}_{\bd\bd} \eqdef
\widehat{\bK}_\text{E} \, \widehat{\bK}_\text{E}^\herm + \sigma_v^2 \, \bI_{\Nr}$.
We have numerically verified that the predominant cause
of BMSE degradation is represented by the first summand
in \eqref{BMSE-LMMSE-1} and, thus, replacing
$\widehat{\bR}_{\bd\bd}$
with $\bR_{\bd\bd}$ in \eqref{BMSE-LMMSE-1} has a very marginal
effect on $\text{BMSE}_\text{LMMSE}^{(1)}$.
Therefore, remembering \eqref{eq:BMMSE-1}, we get
\be
\text{BMSE}_\text{LMMSE}^{(1)} \approx
\text{BMSE}_\text{MMSE}  + \Delta \text{BMSE}_\text{LMMSE}^{(1)}
\ee
with
\be
\Delta \text{BMSE}_\text{LMMSE}^{(1)} \eqdef \Es \left\{
\trace\left[ \bK_\text{B}^\herm \,
\widehat{\bR}_{\by\by}^{-1}
\left(\bR_{\by\by} \, \widehat{\bR}_{\by\by}^{-1} - \bI_{\Nr} \right) \bK_\text{B} \right] \right\} \: .
\label{eq:DeltaBMMSEuno}
\ee
It is apparent from \eqref{eq:DeltaBMMSEuno}
that, in the low-SNR regime, i.e., when $\sigma_v^2$ is sufficiently large
compared to the maximum eigenvalue of $\bK_\text{B} \, \bK_\text{B}^\herm$,
$\bK_\text{E} \, \bK_\text{E}^\herm$,  and
$\widehat{\bK}_\text{E} \, \widehat{\bK}_\text{E}^\herm$,
one has $\widehat{\bR}_{\by\by} \approx \bR_{\by\by} \approx \sigma_v^2 \, \bI_{\Nr}$
and, thus, $\Delta \text{BMSE}_\text{LMMSE}^{(1)} \approx 0$.
On the other hand, for high SNR values, i.e., when $\sigma_v^2$ is sufficiently small
compared to the minimum eigenvalue of $\bK_\text{B} \, \bK_\text{B}^\herm$,
$\bK_\text{E} \, \bK_\text{E}^\herm$, and
$\widehat{\bK}_\text{E} \, \widehat{\bK}_\text{E}^\herm$, it results
that  $\bR_{\by\by} \neq \widehat{\bR}_{\by\by}$,
which implies that $\Delta \text{BMSE}_\text{LMMSE}^{(1)}$
might be nonzero.
In summary, errors regarding the knowledge of the Eve's parameters
may be deleterious at the high-SNR regime, whereas they
are nearly irrelevant for low SNR values.

\subsubsection{Improved LMMSEE}

An alternative design can be pursued by additionally making use of the statistics
of $\Delta\theta_{\text{E},1}, \Delta\theta_{\text{E},2}, \ldots, \Delta\theta_{\text{E},L_\text{E}},\Delta\Ps$. In this case, the structure of the LMMSEE is derived in Appendix~\ref{app:lmmsee}
and it reads as shown in \eqref{eq:LMMSE} at the top of the next page,
\begin{figure*}[!t]
\normalsize
\be
\bhbhatlmmsedue = \bK_\text{B}^\herm \, \left\{\bK_\text{B} \, \bK_\text{B}^\herm +
\Pshat \, \Es\left[e^{-\Delta \Ps}\right]
\frac{1}{L_\text{E}} \sum_{\ell=1}^{L_\text{E}}
\bR_{\ba\ba}^{(\ell)}
+ \sigma_v^2 \, \bI_{\Nr} \right\}^{-1}  \by
\label{eq:LMMSE}
\ee
\hrulefill
\end{figure*}
with
\begin{multline}
\Es\left[e^{-\Delta \Ps}\right] =
\frac{e^{\frac{\sigma_{\Ps}^2}{2}}}{2 \, \text{erf}\left(\frac{\Delta\Psmax}{\sqrt{2} \, \sigma_{\Ps}}\right)} \left[ \text{erf}\left(\frac{\Delta\Psmax-\sigma_{\Ps}^2}{\sqrt{2} \, \sigma_{\Ps}}\right)
\right. \\ \left.
+ \text{erf}\left(\frac{\Delta\Psmax+\sigma_{\Ps}^2}{\sqrt{2} \, \sigma_{\Ps}}\right) \right]
\label{eq:meandeltaPE}
\end{multline}
whereas the $(n_1+1,n_2+1)$th entry of $\bR_{\ba\ba}^{(\ell)}$ is reported
in \eqref{eq:Raa} at the top of the next page,
\begin{figure*}[!t]
\normalsize
\begin{multline}
\left\{\bR_{\ba\ba}^{(\ell)}\right \}_{n_1+1,n_2+1}= \frac{1}{N}
\left \{ 1- \left[4 \pi^2 (n_1-n_2)^2 \Delta^2
\sin^2(\widehat{\theta}_{\text{E},\ell}) - j \, 2 \pi (n_1-n_2)  \Delta
\cos(\widehat{\theta}_{\text{E},\ell}) \right]
\sigma_{\theta_{\text{E}}}^2
\right. \\ \left. \cdot \left[ \frac{1}{2}-
\frac{\Delta \theta_{\text{E}, \text{max}}}{\text{erf}\left(\frac{\Delta \theta_{\text{E}, \text{max}}}{\sqrt{2} \, \sigma_{\theta_{\text{E}}}}\right) \sqrt{2 \pi} \, \sigma_{\theta_{\text{E}}}}
\, e^{- \frac{\Delta \theta_{\text{E}, \text{max}}^2}{2 \sigma_{\theta_{\text{E}}}^2}}
 \right]
\right \} e^{-j [2 \pi (n_1-n_2) \Delta
\cos(\widehat{\theta}_{\text{E},\ell})]}
\label{eq:Raa}
\end{multline}
\hrulefill
\end{figure*}
for $n_1,n_2 \in \{0,1,\ldots, N-1\}$.

The synthesis of \eqref{eq:LMMSE}
essentially requires the same complexity of \eqref{eq:MMSE-S1}
and \eqref{eq:LMMSE-0}.
The estimator \eqref{eq:LMMSE} exploits
the prior information on the
estimation error of the Eve's parameters to
outperform the naive LMMSE,
especially for moderate-to-high SNR values.
Its performance will be numerically evaluated in
the forthcoming Section~\ref{sec:simulation}.

\begin{table}[t]
\caption{System knowledge and computational complexity of the
considered channel estimators
\label{tab: Computational complexity}}
\centering{}%
\begin{tabular}{c||c|c}
\hline
\noalign{\vskip\doublerulesep}
\textbf{Estimator} & \textbf{System Knowledge } & \textbf{Complexity (flops)}\tabularnewline[\doublerulesep]
\hline
\noalign{\vskip\doublerulesep}
\hline
\noalign{\vskip\doublerulesep}
LSE \eqref{eq:LSE} &
$\bK_\text{B}$ (full-column rank)
&
$\mathcal{O}(L_\text{B}^3)$
\tabularnewline[\doublerulesep]
\hline
\noalign{\vskip\doublerulesep}
MLE \eqref{eq:MLE-sol}
&
$\bK_\text{B}$ (full-column rank), $\bR_{\bd\bd}$
&
$\mathcal{O}(\Nr^3)$
\tabularnewline[\doublerulesep]
\hline
\noalign{\vskip\doublerulesep}
MMSEE \eqref{eq:MMSE-S1}
&
$\bK_\text{B}$, $\bR_{\by\by}$
&
$\mathcal{O}(\Nr^3)$
\tabularnewline[\doublerulesep]
\hline
\noalign{\vskip\doublerulesep}
LMMSEE \eqref{eq:LMMSE-0}
&
$\bK_\text{B}$,
$\{\widehat{\theta}_{\text{E},\ell}\}_{\ell=1}^{L_\text{E}}$, $\Pshat$, $\sigma_v^2$
&
$\mathcal{O}(\Nr^3)$
\tabularnewline[\doublerulesep]
\hline
\noalign{\vskip\doublerulesep}
LMMSEE \eqref{eq:LMMSE}
&
$\bK_\text{B}$,
$\{\widehat{\theta}_{\text{E},\ell}\}_{\ell=1}^{L_\text{E}}$, $\Pshat$, $\sigma_v^2$,
\\ &
pdf of
$\{\Delta\theta_{\text{E},\ell}\}_{\ell=1}^{L_\text{E}}$ and $\Delta\Ps$
&
$\mathcal{O}(\Nr^3)$
\tabularnewline[\doublerulesep]
\hline
\hline
\end{tabular}
\end{table}

\section{Numerical results}
\label{sec:simulation}

Tab.~\ref{tab: Computational complexity} reports
the system information required for calculating
the considered estimators and the corresponding
computation complexity.
The performance analysis of such estimators was developed by
resorting to Monte Carlo simulations in order to
corroborate our theoretical findings  as well. To this aim, we considered
the following simulation setting.
The number of antennas at Alice is set equal to $\Nr=10$,
with an absolute antenna spacing $d=\lambda_\text{c}/2$.
With reference to the multi-user scenario depicted in Fig.~\ref{fig:fig_1},
we considered two Bob-Eve pairs, i.e., $M=2$.
The number of Bob-to-Alice and Eve-to-Alice paths  was chosen equal to
$L_\text{B,1}=L_{\text{E},1}=3$ for the first Bob-Eve pair, whereas
$L_\text{B,2}=L_{\text{E},2}=2$  for the second one.
The AoAs of $\text{Bob} \,\,1$ and $\text{Bob} \,\, 2$ were fixed as follows:
$\theta_{\text{B},1,1}=0$,
$\theta_{\text{B},2,1}=\psi$,
$\theta_{\text{B},3,1}=\pi/5$,
$\theta_{\text{B},1,2}=(3/5) \pi$, and
$\theta_{\text{B},2,2}=(7/10) \pi$,
respectively, with the parameter
$\psi \in [0, \pi/10]$.
It should be observed that,
when $\psi \to 0$, the matrix $\bK_{\text{B},1}$
tends to lose its full column rank property.
On the other hand, the AoAs of $\text{Eve} \,\,1$ and $\text{Eve} \,\, 2$
were chosen as:
$\theta_{\text{E},1,1}=\pi/5 + \phi$,
$\theta_{\text{E},2,1}=(2/5) \pi$,
$\theta_{\text{E},3,1}=\pi/2$,
$\theta_{\text{E},1,2}=(4/5) \pi$,  and
$\theta_{\text{E},2,2}=(9/10) \pi$, respectively,
with the parameter
$\phi \in [0, \pi/10]$. It is noteworthy that, when
$\phi \to 0$, the columns of
$\bA_{\text{B},1}$ and $\bA_{\text{E},1}$
become linearly dependent.
The pilot vectors $\bp_1$ and $\bp_2$ were obtained
by picking two different columns of an unitary
$K$-point discrete Fourier transform matrix, with $K=8$.
Unless otherwise specified, we set  $\psi=\phi=\pi/10$,
$\text{SNR}_\text{B} \eqdef \EuScript{P}_{\text{B},1}/\sigma_v^2=
\EuScript{P}_{\text{B},2}/\sigma_v^2=30$ dB, $\text{SSR} \eqdef
\EuScript{P}_{\text{B},1}/\EuScript{P}_{\text{E},1}=
\EuScript{P}_{\text{B},2}/\EuScript{P}_{\text{E},2}
= 0$ dB,
$\sigma_{\theta_{\text{E}}}=\Delta \theta_{\text{E}, \text{max}}/3$,
$\Delta \theta_{\text{E}, \text{max}}=\pi/25$
(corresponding to a interval of uncertainty
$2 \, \Delta \theta_{\text{E}, \text{max}}$ of $0.08 \pi$ rad),
$\sigma_{\Ps}=\Delta\Psmax/2$, $\Delta\Psmax=0.3454$
(corresponding to a interval of uncertainty
$2 \, \Delta\Psmax$ of $3$ dB),
and
we implemented the MLE and MMSEE \eqref{eq:MMSE-S1}
by using the exact
expression of $\bR_{\bd\bd}$ and $\bR_{\by\by}$.

Two performance metrics were used to evaluate the channel estimation
performance of the first legitimate user. The former is a \textit{normalized}
version of the BMSE
defined in \eqref{eq:AMSE}:
\be
\text{NBMSE}(\widehat{\bh}_{\text{B},1}) \eqdef \Es \left[ \frac{\left \|
\widehat{\bh}_{\text{B},1} -\bh_{\text{B},1}
\right\|^2}{\|\bh_{\text{B},1}\|^2} \right] \: .
\ee
The latter is the achievable \textit{(ergodic) secrecy rate} \cite{Leung.1978}
of the downlink transmission from Alice to Bob $1$:
\be
\Cap_\text{s} \eqdef \max(\Cap_{\text{B},1}-\Cap_{\text{E},1}, 0)
\ee
where, for $\text{RX} \in \{\text{B}, \text{E}\}$,
\be
\Cap_{\text{RX},1}  =\Es \left[ \log_2 \left( 1+ \text{SINR}_{\text{RX},1} \right) \right]
\ee
is the maximum achievable \textit{normalized}\footnote{The normalization
by the factor $T_\text{d}/T_\text{c}$ was introduced for convenience, where
we remember that
$T_\text{d}$ is the length of the downlink
data session and $T_\text{c}$ is the channel coherence time.
}
(ergodic) spectral efficiency
(in $\text{bits}/\text{s}/\text{Hz}$) of the
Alice-to-RX downlink channel,
\be
\text{SINR}_{\text{RX},1} \eqdef \frac{\text{SNR}_\text{DL}
\, \left |\bh_{\text{RX},1}^\trasp \, \bA_{\text{RX},1}^\trasp \, \bw_1 \right|^2}
{\text{SNR}_\text{DL} \displaystyle \sum_{m=2}^{M}
\left |\bh_{\text{RX},1}^\trasp \, \bA_{\text{RX},1}^\trasp \, \bw_m \right|^2 + 1}
\ee
denotes the corresponding signal-to-interference-plus-noise ratio (SINR)
under the assumption that the Alice transmits independent and identically
distributed zero-mean unit-variance symbols,
with $\text{SNR}_\text{DL}$ representing the SNR (assumed to be independent of $m$) and
$\bW \eqdef [\bw_1, \bw_2, \ldots, \bw_M] \in \C^{N \times M}$ being the 
precoding matrix at Alice.
We considered the unit-norm matched-filter precoder \cite{Joham.2005}, which is given by $\bW =\widehat{\bH}_{\text{B}}^*/\|\widehat{\bH}_{\text{B}}\|$, with
$\widehat{\bH}_{\text{B}} \eqdef [\bA_{\text{B},1} \, \widehat{\bh}_{\text{B},1},
\bA_{\text{B},2} \, \widehat{\bh}_{\text{B},2}, \ldots,
\bA_{\text{B},M} \, \widehat{\bh}_{\text{B},M}]$.
As a reference, we also reported the performance of
the estimators \eqref{eq:LSE}, \eqref{eq:MLE-sol}, and \eqref{eq:MMSE-S1}
when the Eves do not attack the pilot
session of the legitimate users, i.e., $\Psm=0$, $\forall m \in \mathcal{M}$,
and, thus, they steal information in downlink only, referred to as ``passive Eves".
In this respect, it should be observed that $\bR_{\bd\bd} =\sigma_v^2 \, \bI_{\Nr}$
and, thus, the MLE \eqref{eq:MLE-sol} ends up to the LSE \eqref{eq:LSE}.
Finally, all the results are obtained by carrying out
$10^5$ independent Monte Carlo trials, with each run using
different sets of channel coefficients and noise.

\subsection{Example~1 : Performance as a function of $\text{SNR}_\text{B}$}

In this subsection, we reported the performance of the consi\-dered channel
estimators as a function of $\text{SNR}_\text{B}$,
ranging from $0$ to $34$ dB.
Results of Fig.~\ref{fig:fig_3} confirm that
the LSE is unable to counteract the pilot spoofing attack, by showing
the BMSE floor predicted by Lemma~\ref{lem:1}. On the other hand,
according to Lemmas~\ref{lem:2} and \ref{lem:3},
the MLE and MMSEE are able to suppress the pilot spoofing signal in the
high-$\text{SNR}_\text{B}$ region, by exhibiting almost the same performance
of the corresponding estimators in the absence of the pilot spoofing attack.
The performance of the
naive LMMSE gets worse for increasing values of $\text{SNR}_\text{B}$, due
to the presence of the summand
$\Delta \text{BMSE}_\text{LMMSE}^{(1)}$ in
\eqref{eq:DeltaBMMSEuno}.
Such a negative effect is compensated for by exploiting the knowledge of
the estimation error of the Eves' parameters, as testified by the satisfactory
asymptotic (i.e., for $\text{SNR}_\text{B} \to + \infty$) performance of the improved LMMSE.

\vspace{-3mm}
\subsection{Example~2: Performance as a function of $\text{SSR}$}

We depicted in Fig.~\ref{fig:fig_4} the performance of the considered channel
estimators as a function of the SSR,
ranging from -$10$ to $20$ dB. Besides confirming that, for high SNR values,
the performance of the MLE and MMSEE is almost unaffected by the pilot
spoofing attack, independently of the value of the SSR, it is apparent that the
channel estimation accuracy of the LSE becomes acceptable only for high SSR values.
It is also interesting to note that the performance of the naive LMMSE
rapidly worsens as the SSR decreases, while the improved LMMSE exhibits
a spoofing-resistant capability for a wider range of SSR values.

\subsection{Example~3: Performance of the MMSEE
with sample correlation matrix $\bS_{\by\by}$}
\label{sub_ex3}

To show how much training is needed to reliably estimate
$\bR_{\by\by}$ according to the training protocol reported
in Fig.~\ref{fig:fig_2}, we plotted in Fig.~\ref{fig:fig_5} the performance
of the MMSEE as a function of $Q$.
As expected, the SMI implementation \eqref{eq:MMSE-S1-smi} requires
a huge number $Q$ of channel coherence intervals for achieving satisfactory
performance, whereas its subspace counterpart
\eqref{eq:MMSE-S1-sub} converges to the ideal BMSE
\eqref{eq:BMMSE-1} much more quickly, by ensuring
a reduction of the length of the time window $\mathcal{W}$ of about one order of magnitude. Similar conclusions apply to the MLE as well, when it is implemented
starting from $\bS_{\bd\bd}$.

\subsection{Example~4: Performance as a function of $\psi$ and $\phi$}

We depicted in Fig.~\ref{fig:fig_8} the performance of the considered channel
estimators as a function of $\psi$.
When $\psi \to 0$, the condition $\rank(\bK_\text{B})=L_\text{B}$ tends to be violated.
The non-fulfillment of  such a rank condition
does not prevent the MLE, MMSEE, and LMMSEEs to satisfactorily estimate
the legitimate channel, although perfect cancellation of the pilot spoofing
signal at high SNR values is not ensured anymore.

We also studied the impact of $\phi$
on the performance of the considered channel
estimators. Results of Fig.~\ref{fig:fig_9} show that
the MLE, MMSEE, and LMMSEEs exhibit a certain robustness
when the legitimate and the spoofing transmissions tend to have
a common AoA, i.e., when  $\phi \to 0$, and, thus,
the columns of
$\bA_\text{B}$ and $\bA_\text{E}$ are no longer
linearly independent.

\subsection{Example~5: Performance as a function of $\sigma_{\theta_{\text{E}}}$}

Finally, we investigated the performance of the LMMSEEs derived in
Subsection~\ref{sec:MMSE-2} as a function of the
standard deviation $\sigma_{\theta_{\text{E}}}$,
with $\Delta \theta_{\text{E}, \text{max}}=\pi/(25)$.\footnote{Results -- not
reported here for the sake of brevity -- show that the performance of
the LMMSEEs are weakly affected by the error on the
estimate of the Eve's average transmission power $\Ps$.}
It can be seen from Fig.~\ref{fig:fig_10} that the performance
of the LMMSEEs gracefully degrades as
the uncertainty on the knowledge
of the spoofing AoAs increases, by enlightening that even
an imperfect knowledge of the spoofing transmission parameters
can lead to a significant performance gain with respect to the simpler LSE.

\begin{figure*}[!t]
\begin{minipage}[b]{9cm}
\centering
\includegraphics[width=1\linewidth]{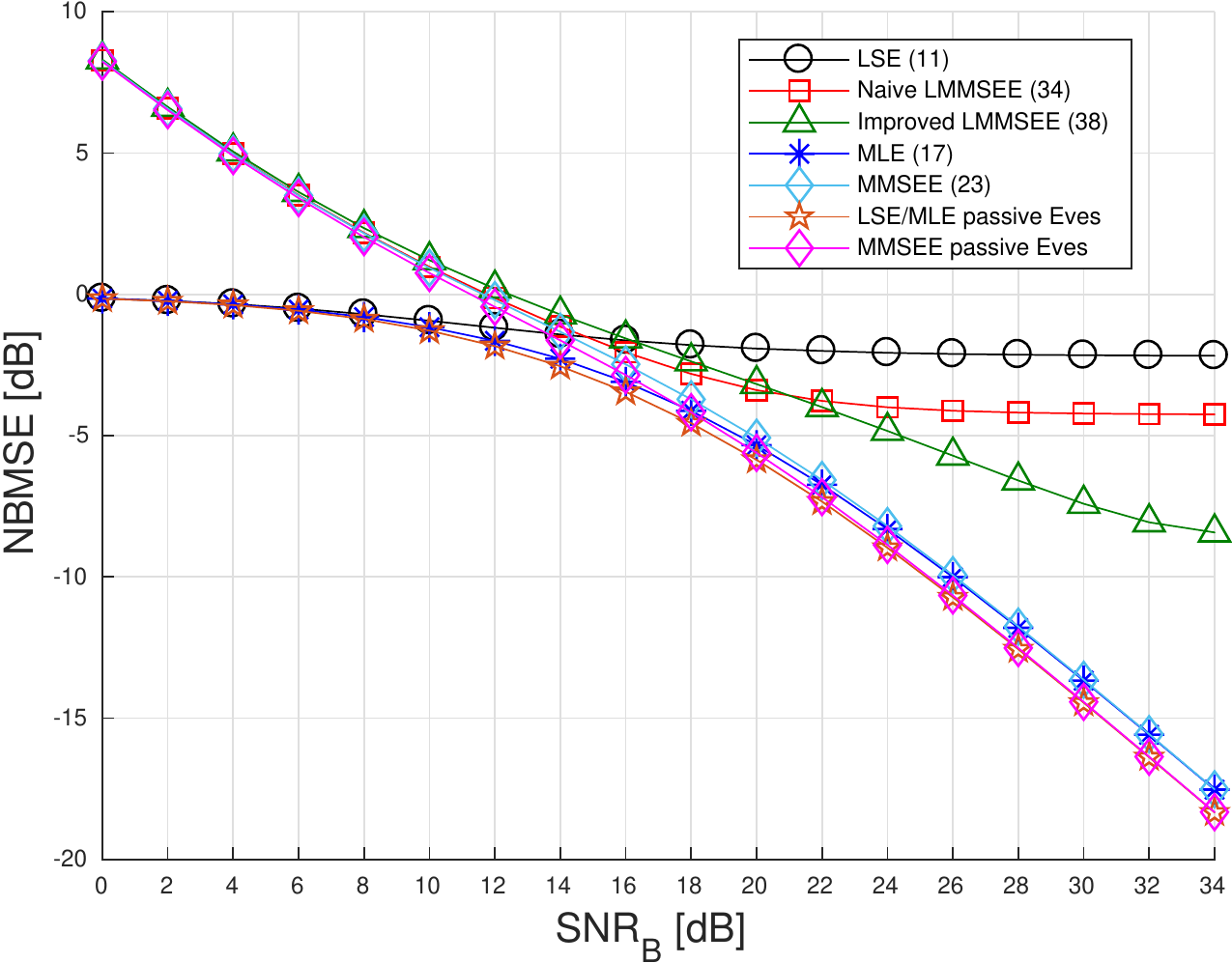}
\end{minipage}
\begin{minipage}[b]{9cm}
\centering
\includegraphics[width=1\linewidth]{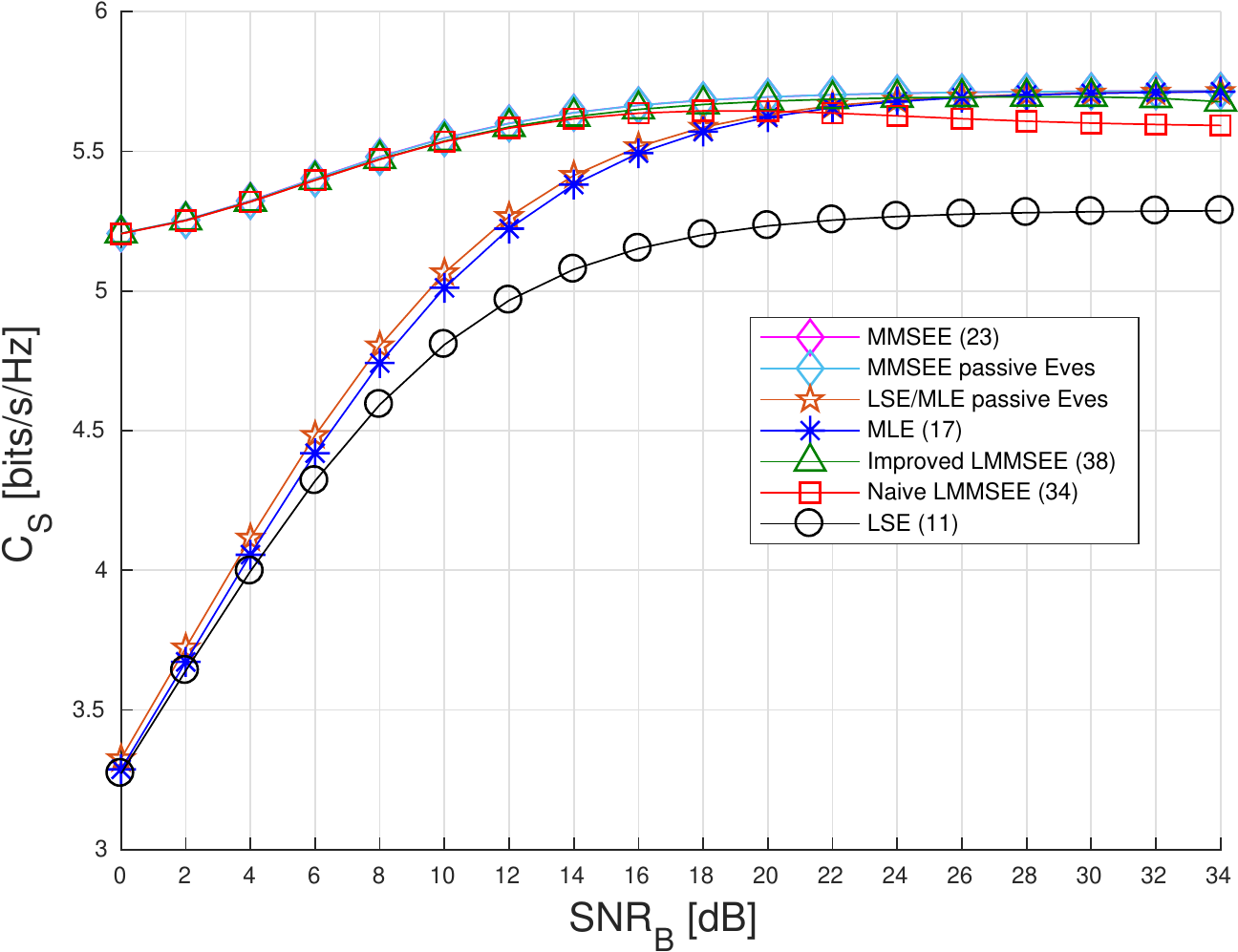}
\end{minipage}
\caption{NBMSE (left) and secrecy rate (right) versus $\text{SNR}_\text{B}$ (Example~1).}
\label{fig:fig_3}
\end{figure*}

\begin{figure*}[!t]
\begin{minipage}[b]{9cm}
\centering
\includegraphics[width=1\linewidth]{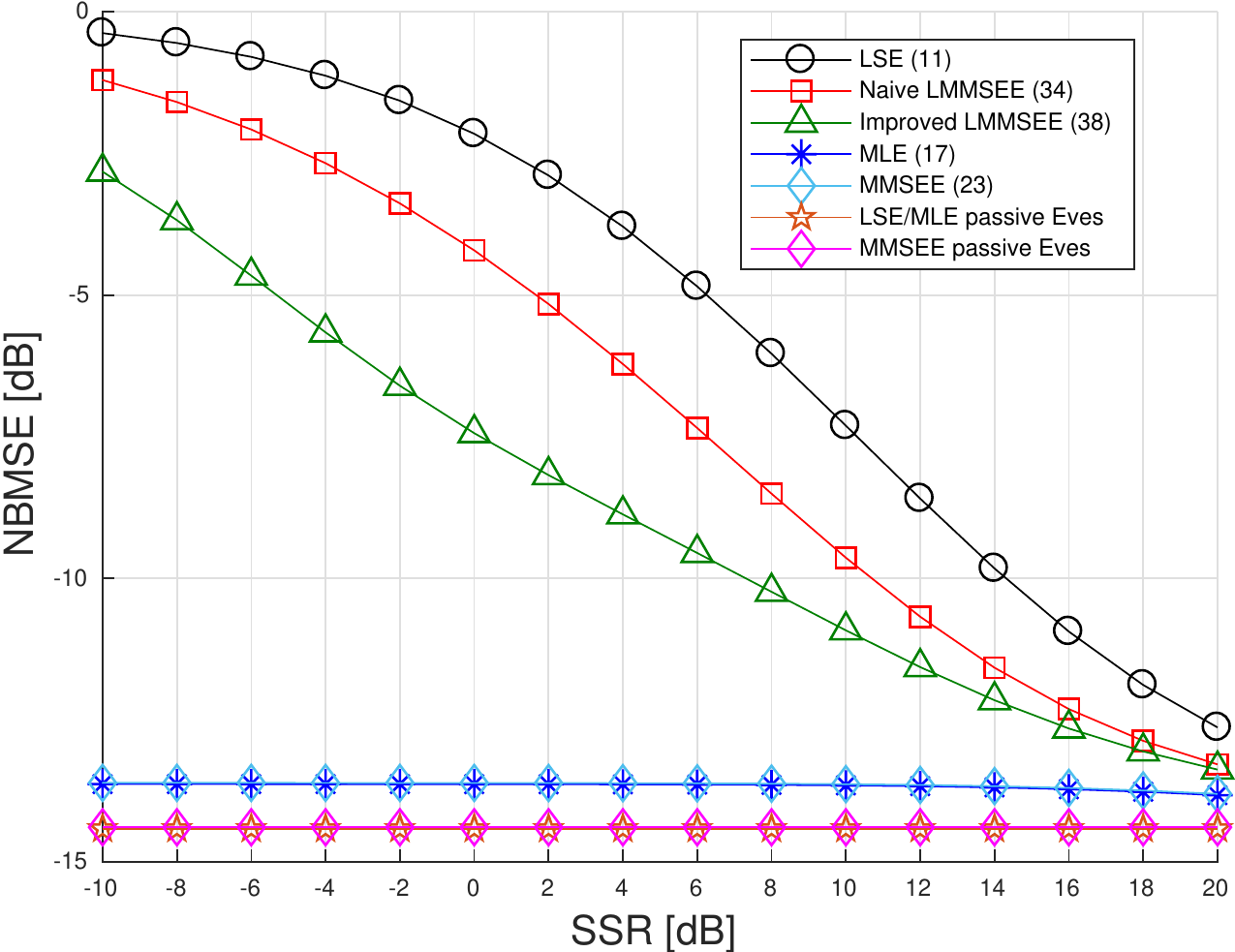}
\end{minipage}
\begin{minipage}[b]{9cm}
\centering
\includegraphics[width=1\linewidth]{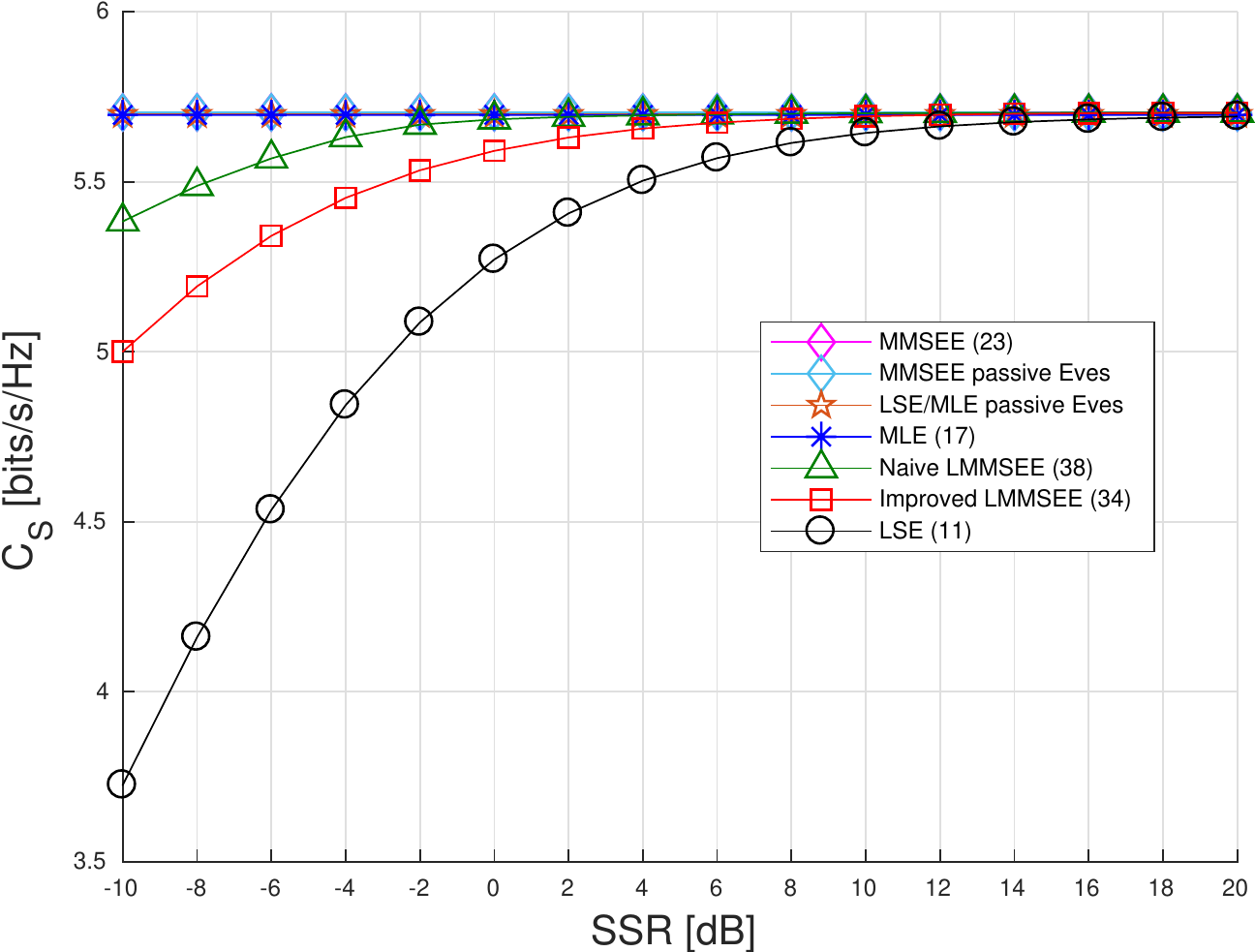}
\end{minipage}
\caption{NBMSE (left) and secrecy rate (right) versus $\text{SSR}$ (Example~2).}
\label{fig:fig_4}
\end{figure*}

\begin{figure}
\centering
\includegraphics[width=1\columnwidth]{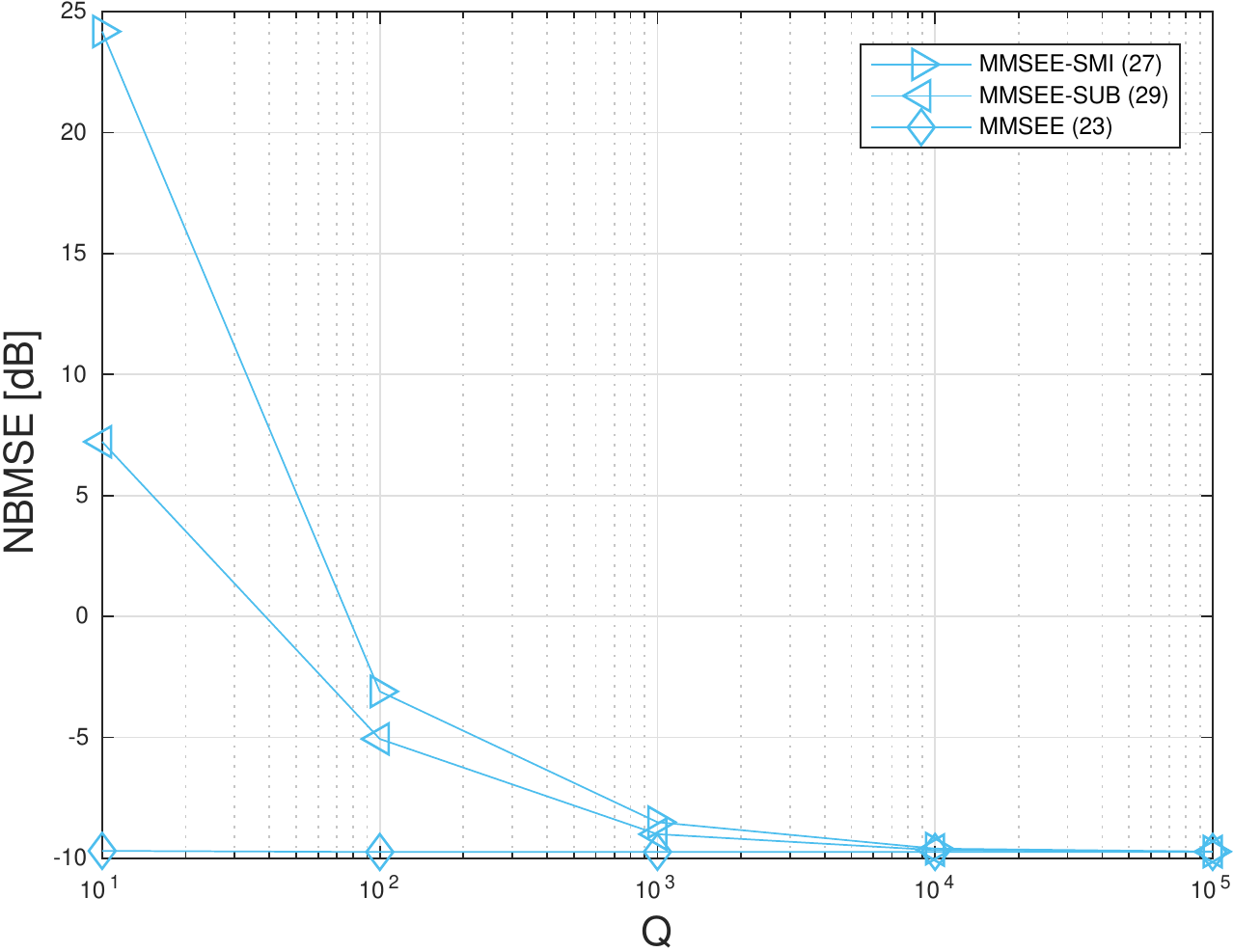}
\caption{NBMSE versus $Q$ (Example~3).}
\label{fig:fig_5}
\end{figure}

\begin{figure*}[!t]
\begin{minipage}[b]{9cm}
\centering
\includegraphics[width=1\linewidth]{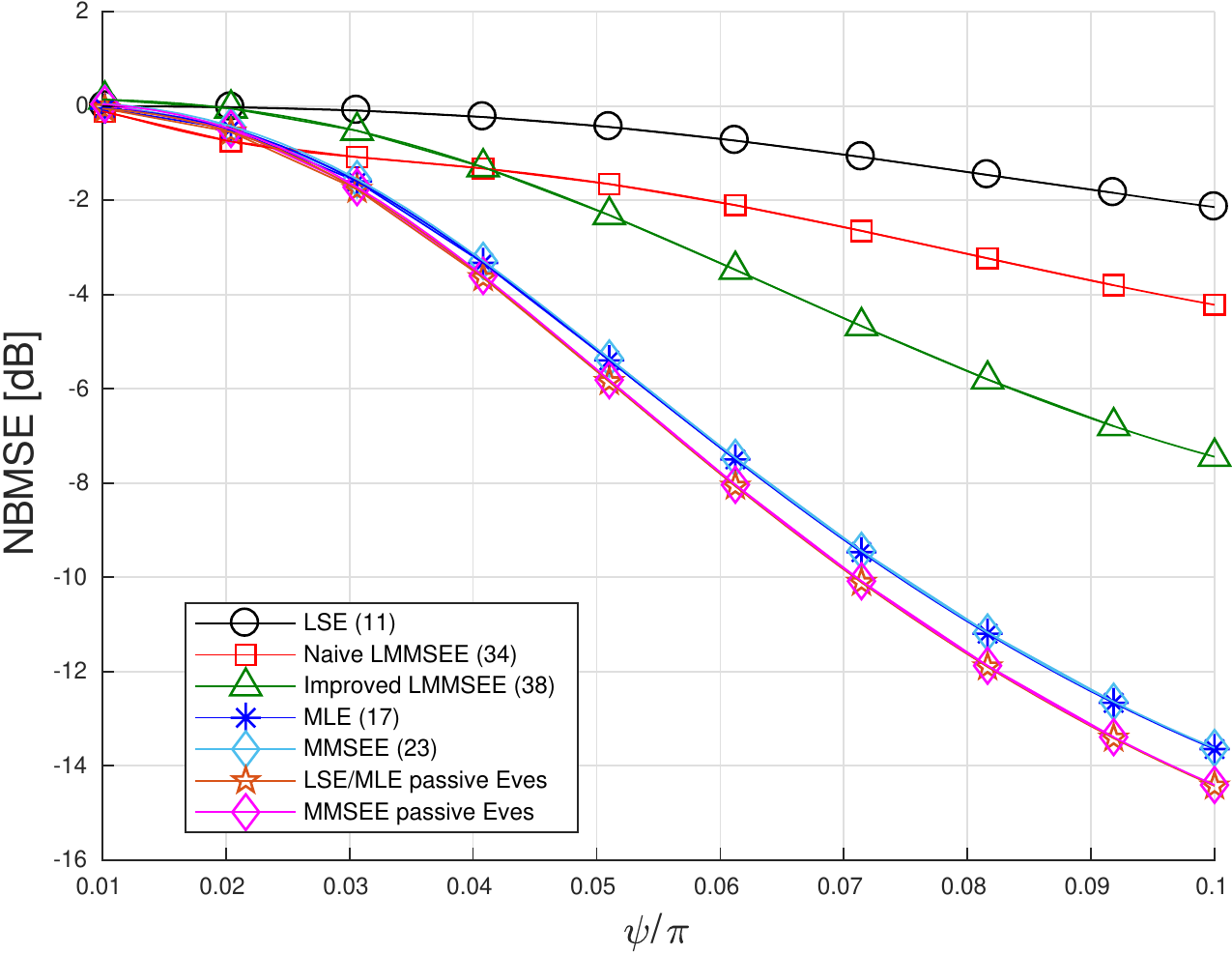}
\end{minipage}
\begin{minipage}[b]{9cm}
\centering
\includegraphics[width=1\linewidth]{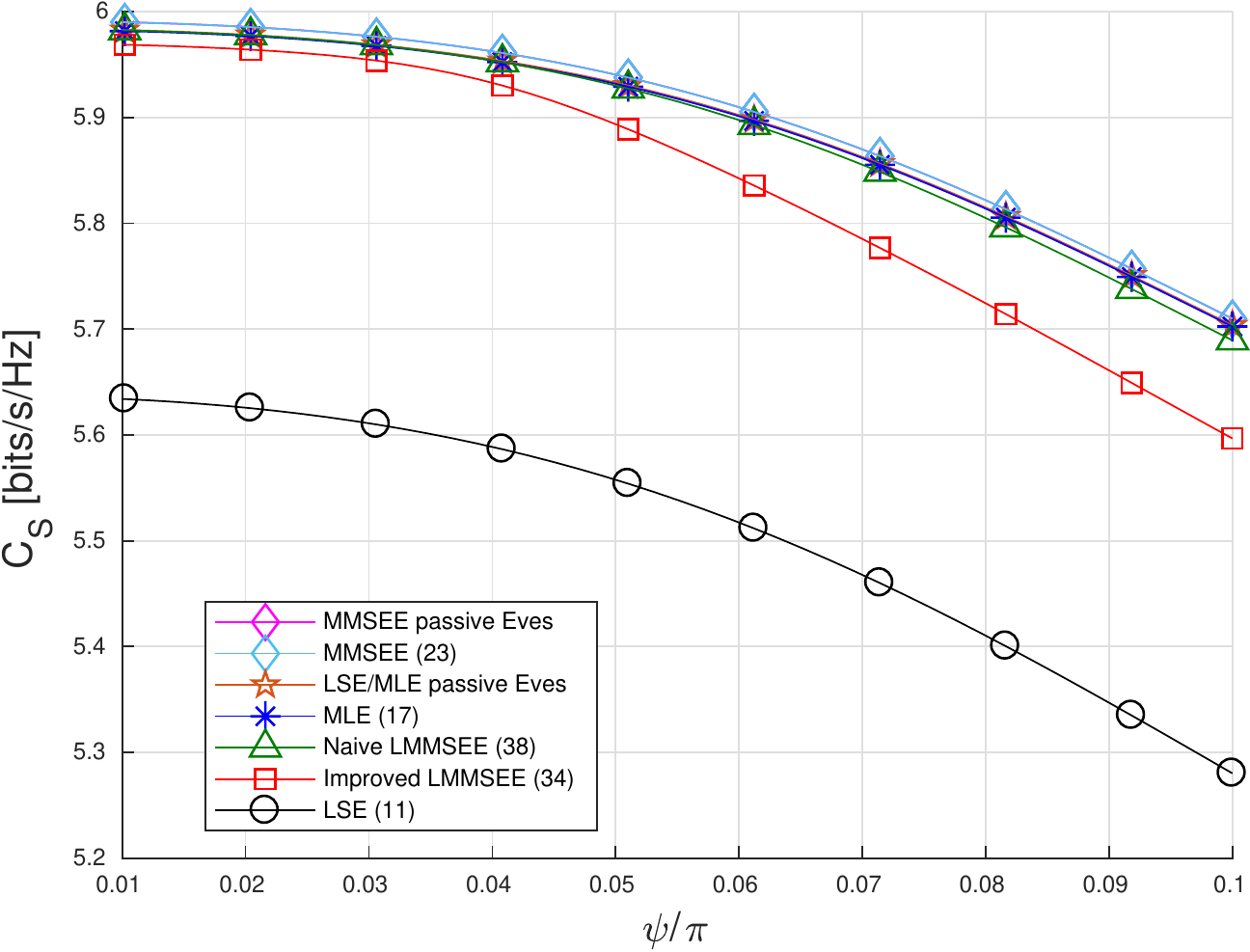}
\end{minipage}
\caption{NBMSE (left) and secrecy rate (right) versus $\psi$ (Example~4).}
\label{fig:fig_8}
\end{figure*}
\begin{figure*}[!t]
\begin{minipage}[b]{9cm}
\centering
\includegraphics[width=1\linewidth]{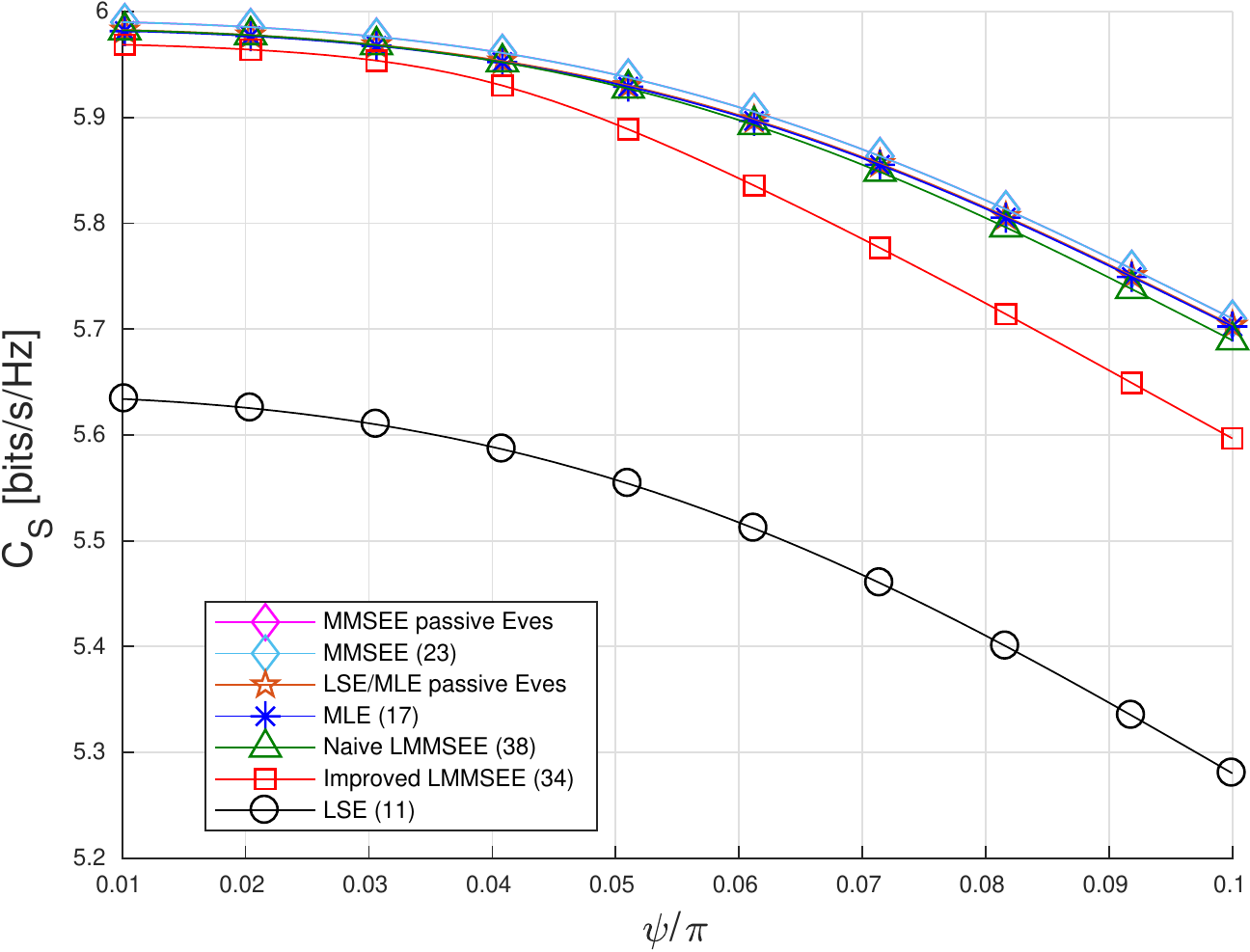}
\end{minipage}
\begin{minipage}[b]{9cm}
\centering
\includegraphics[width=1\linewidth]{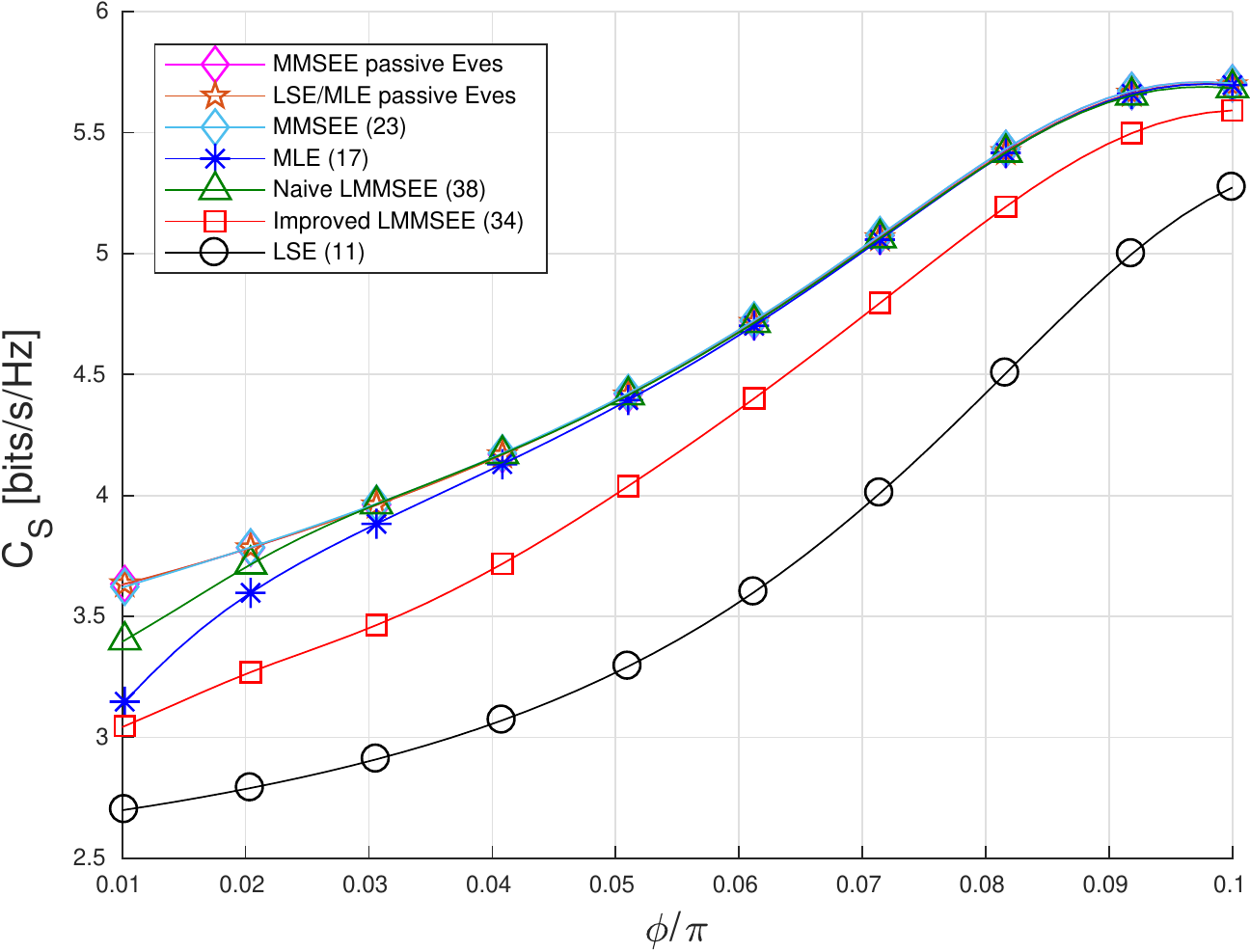}
\end{minipage}
\caption{NBMSE (left) and secrecy rate (right) versus $\phi$ (Example~4).}
\label{fig:fig_9}
\end{figure*}

\begin{figure*}[!t]
\begin{minipage}[b]{9cm}
\centering
\includegraphics[width=1\linewidth]{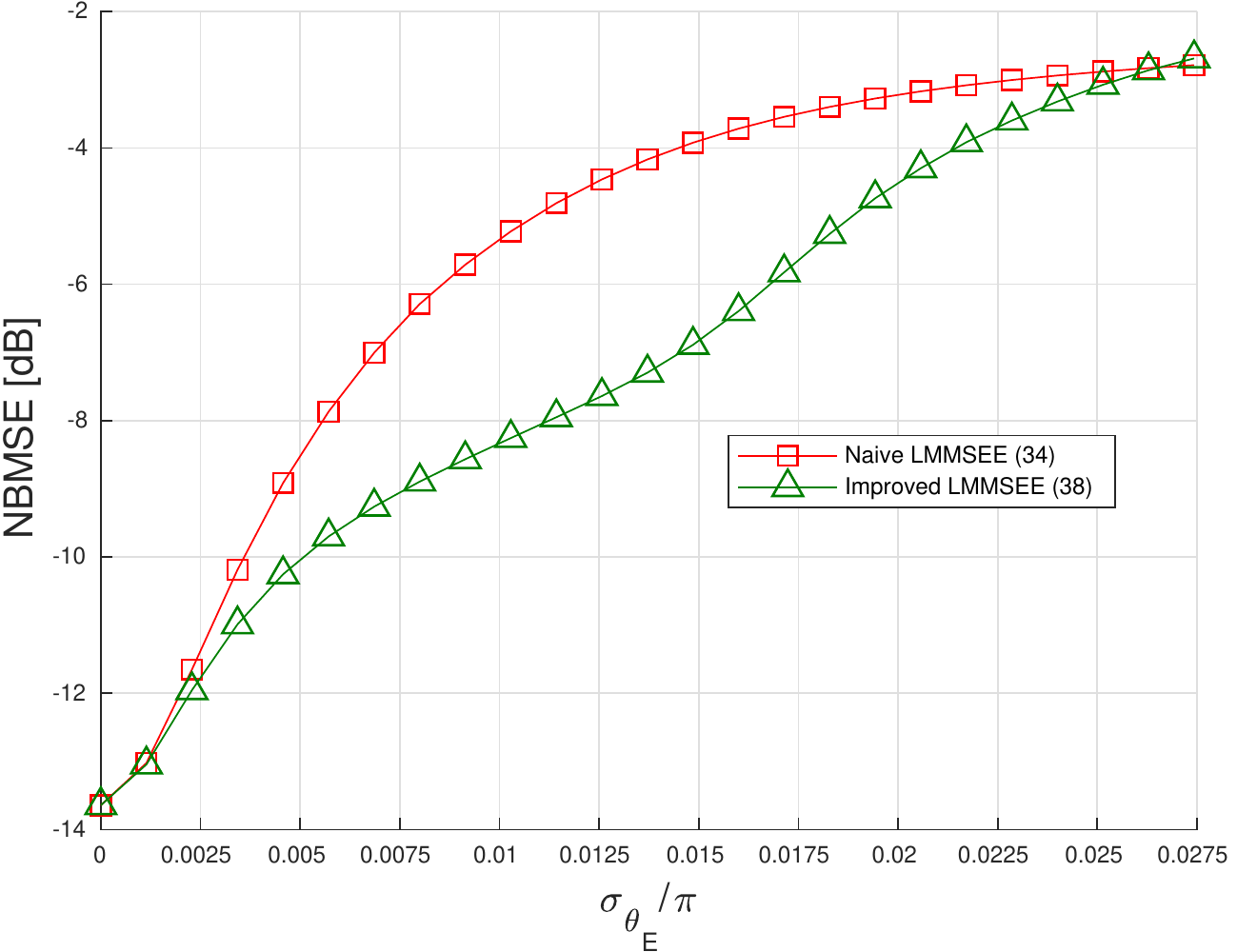}
\end{minipage}
\begin{minipage}[b]{9cm}
\centering
\includegraphics[width=1\linewidth]{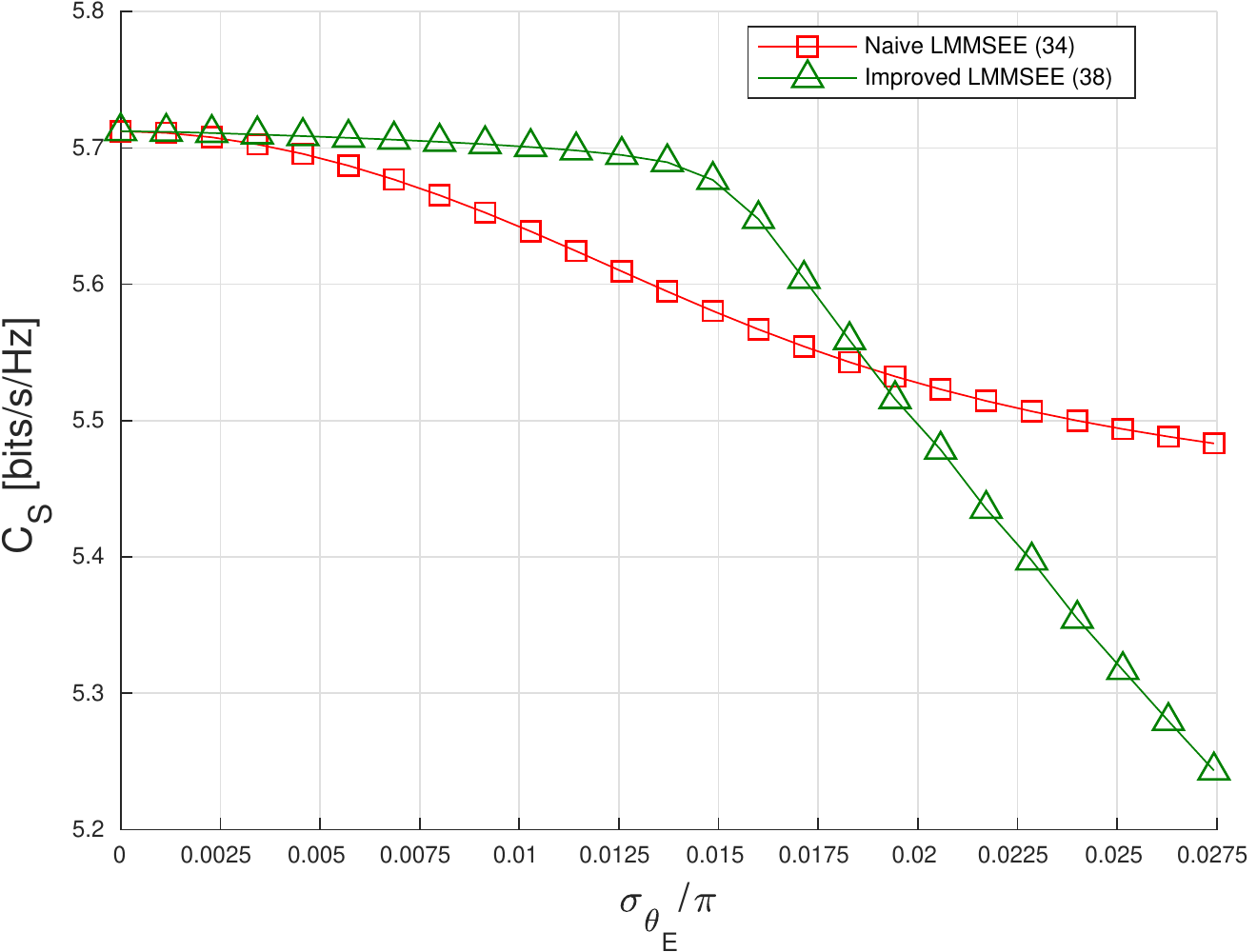}
\end{minipage}
\caption{NBMSE (left) and secrecy rate (right) versus $\sigma_{\theta_{\text{E}}}$ (Example~5).}
\label{fig:fig_10}
\end{figure*}

\section{Conclusions and directions for future work}
\label{sec:concl}

Five uplink channel estimation schemes for multiple antennas systems
have been developed and studied in the case of a pilot spoofing attack, namely,
LSE, MLE, MMSEE, naive and improved LMMSEEs.
The LSE does not require knowledge of 
the statistics of the legitimate channel 
or the correlation matrix
of the spoofing-plus-noise signal and, hence, it
does not have any
spoofing suppression capability.
On the other hand, compared to the LSE, the MLE has better accuracy
as it involves the correlation matrix of the spoofing-plus-noise signal.
At low SNR, a performance gain over the MLE is ensured by the MMSEE, 
since it additionally incorporates the statistics of the legitimate channel.
The naive LMMSEE can be designed if an estimate
of the main spoofing parameters is available, whereas the improved LMMSEE also exploits
prior information regarding the estimation error.

Specifically, the following main results have been found:

\begin{enumerate}[i)]

\item
Under certain operative conditions, the MLE and MMSEE can be capable of
perfectly rejecting the pilot  spoofing signal in the high-SNR regime.

\item
Both the MLE and MMSEE can be entirely implemented from data, but
a training session spanning
multiple channel coherence intervals is required in a setup phase.

\item
The estimation error of the spoofing AoAs mainly affects
the performance of the naive LMMSEE at high SNR.

\item
The improved LMMSEE largely outperforms the LSE and
the naive LMMSEE even for low SSR values.

\end{enumerate}

In summary, this study demonstrates that,
if more sophisticated channel estimators
than the LSE are employed, an uplink
pilot spoofing attack might be
effectively counteracted at the base station,
resulting in a very limited signal leakage to Eve
in the downlink data phase.
Finally,  we assumed that a  ULA  is used
at  the  BS and focused on single-antenna legitimate users.
To improve physical-layer security, a viable strategy
is to exploit additional spatial dimensions. In this
respect, a first interesting
research subject consists of considering
three-dimensional (or  full-dimensional)  MIMO
at the BS and user terminals equipped with
multiple antennas.
Moreover, the effects of uplink channel estimation errors
on the downlink secrecy rate were studied through
numerical simulations. An additional research issue is
to develop a theoretical analy\-sis of the downlink SINR
that explicitly accounts for both correlation matrix and
channel estimation effects.
When the legitimate users employ non-orthogonal
pilot sequences,  it will be also interesting to study the joint effects
of pilot spoofing attacks and pilot contamination
on uplink channel estimation and downlink
secrecy rate.

%%%%%%%%%%%%%%%%%%Appendices%%%%%%%%%%%%%%%%%%%

\appendices

\section{Proof of Lemma~\ref{lem:1}}
\label{app:lem-1}

First, we observe that $\bA_\text{B} \left(\bA_\text{B}^\herm \, \bA_\text{B}\right)^{-2}
\bA_\text{B}^\herm$ and $\bA_\text{E} \, \bA_\text{E}^\herm$
are positive semi\-definite Hermitian matrices. For positive semi\-definite Hermitian matrices
$\bA \in \C^{n \times n}$ and
$\bB \in \C^{n \times n}$, with eigenvalues
sorted decreasingly
$a_1 \ge a_2 \ge \cdots \ge a_n$ and
$b_1 \ge b_2 \ge \cdots \ge b_n$, respectively,
it results \cite{Marshall} that
\be
\sum_{i=1}^n a_i \, b_{n-i+1} \le \trace(\bA \, \bB) \le \sum_{i=1}^n a_i \, b_i \: .
\ee
Therefore, the bound \eqref{eq:AMSE-LSE-asympt-bound} comes from
recalling the facts \cite{Horn.book.1990} that:
(i) the two matrices $\bA_\text{B} \left(\bA_\text{B}^\herm \, \bA_\text{B}\right)^{-2} \bA_\text{B}^\herm$ and $\left(\bA_\text{B}^\herm \, \bA_\text{B}\right)^{-1}$ have
the same nonzero eigenvalues;
(ii) the singular values $\sigma_{\ell}(\bA_\text{B})$ and
$\sigma_{\ell}(\bA_\text{E})$ are the nonnegative square roots
of the eigenvalues of $\bA_\text{B} \, \bA_\text{B}^\herm$
and $\bA_\text{E} \, \bA_\text{E}^\herm$, respectively.

\section{Proof of Lemma~\ref{lem:2}}
\label{app:lem-2}

By virtue of the matrix inversion lemma \cite{Horn.book.1991}, one has
\barr
\bR_{\bd\bd}^{-1} & =
\left(\bK_\text{E} \, \bK_\text{E}^\herm + \sigma_v^2 \, \bI_{\Nr}\right)^{-1}
\nonumber \\ & =
\frac{1}{\sigma_v^2} \left[ \bI_{\Nr}- \bK_\text{E} \left( \bK_\text{E}^\herm \, \bK_\text{E} + \sigma_v^2 \, \bI_{L_\text{E}}\right)^{-1} \bK_\text{E}^\herm \right] \: .
\label{eq:app2-1}
\earr
By substituting \eqref{eq:app2-1} in \eqref{eq:AMSE-ML}
and resorting again to the matrix inversion lemma, one gets
\begin{multline}
\text{BMSE}_\text{ML} = \sigma_v^2 \, \trace\left[ \left(\bK_\text{B}^\herm  \, \bK_\text{B}\right)^{-1} \right]
\\ + \sigma_v^2 \, \trace\left[ \bK_\text{B}^\dag \, \bK_\text{E}
\left(\bK_\text{E}^\herm \, \bP_\text{B} \, \bK_\text{E} + \sigma_v^2 \, \bI_{L_\text{E}} \right)^{-1} \bK_\text{E}^\herm \left(\bK_\text{B}^\herm\right)^\dag\right]
\label{eq:app2-2}
\end{multline}
where $\bP_\text{B} \eqdef \bI_{\Nr} - \bK_\text{B} \left(\bK_\text{B}^\herm  \, \bK_\text{B}\right)^{-1} \bK_\text{B}^\herm \in \C^{\Nr \times \Nr}$ is the orthogonal projector
onto $\nullo(\bK_\text{B}^\herm)$.
We recall that $\mathcal{N}(\bK_\text{B}^\herm)$ is
the orthogonal complement of $\range(\bK_\text{B})$ in $\C^{\Nr}$.

It is apparent from \eqref{eq:app2-2} that, if $\bK_\text{E}^\herm \, \bP_\text{B} \, \bK_\text{E}=\bO_{L_\text{E} \times L_\text{E}}$, i.e.,
$\range(\bK_\text{E}) \subseteq \range(\bK_\text{B})$, then
\be
\overline{\text{BMSE}}_\text{ML} =
\trace\left[ \bK_\text{B}^\dag \, \bK_\text{E} \,
\bK_\text{E}^\herm \left(\bK_\text{B}^\herm\right)^\dag\right] \neq 0
\ee
that is, in the absence of noise, the BMSE of the MLE
exhibits a saturation effect due to the concurrent spoofing transmission.
On the other hand, if
$\range(\bK_\text{E}) \subseteq \Orange(\bK_\text{B})$
or, equivalently,
$\range(\bK_\text{B}) \cap \range(\bK_\text{E})=\{\mathbf{0}_{K \Nr}\}$,
one has $\bK_\text{B}^\dag \, \bK_\text{E}=\bO_{L_\text{B} \times L_\text{E}}$
and, hence, $\overline{\text{AMSE}}_\text{ML}=0$: in this case,
in the absence of noise, perfect spoofing suppression is achieved.
As a final step, we observe tha condition
$\range(\bK_\text{B}) \cap \range(\bK_\text{E})=\{\mathbf{0}_{\Nr}\}$
is equivalent to
$\range(\bA_\text{B}) \cap \range(\bA_\text{E})=\{\mathbf{0}_{\Nr}\}$.

\section{Case 1: Evaluation of the Bayesian CRLB}
\label{app:FIM-1}

Let $p(\by,\bhb)$ be the \textit{joint} pdf
of $\by$ and $\bhb$,
under regularity conditions that are satisfied by Gaussian
random vectors \cite{VanTrees,Dong.2002}, the {\em Bayesian information matrix (BIM)}
is defined as
\be
\Fcal \eqdef \Es_{\by,\bhb}\left\{
\frac{\partial \ln p(\by,\bhb)}{\partial \bhb^*}
\left[\frac{\partial \ln p(\by,\bhb)}{\partial \bhb^*}\right]^\herm
\right\} \in \C^{L_\text{B} \times L_\text{B}} \: .
\label{eq:FIM-1}
\ee
The Bayesian CRLB is given by $\text{BCRLB}(\bhb)=\trace(\Fcal^{-1})$.

By resorting to the conditional expectation rule \cite{Casella}, the BIM
\eqref{eq:FIM-1} can be equivalently written as
\be
\Fcal = \Es_{\bhb} \left\{ \Es_{\by|\bhb}\left\{
\frac{\partial \ln p(\by,\bhb)}{\partial \bhb^*}
\left[\frac{\partial \ln p(\by,\bhb)}{\partial \bhb^*}\right]^\herm
\, \Big | \, \bhb \right\} \right\}
\label{eq:app1_1}
\ee
Since $p(\by,\bhb) = p(\by \, | \, \bhb) \, p(\bhb)$, the
second expectation in \eqref{eq:app1_1} becomes
\begin{multline}
 \Es_{\by|\bhb}\left\{
\frac{\partial \ln p(\by,\bhb)}{\partial \bhb^*}
\left[\frac{\partial \ln p(\by,\bhb)}{\partial \bhb^*}\right]^\herm
\, \Big | \, \bhb \right\}
\\ = \Fcal_{\text{c}} +
\frac{\partial \ln p(\bhb)}{\partial \bhb^*} \left[\frac{\partial \ln p(\bhb)}{\partial \bhb^*}\right]^\herm
\label{eq:app1_2}
\end{multline}
with
\be
\Fcal_{\text{c}} \eqdef
\Es_{\by|\bhb} \left\{
\frac{\partial \ln p(\by \, | \, \bhb)}{\partial \bhb^*}
\left[\frac{\partial \ln p(\by \, | \, \bhb)}{\partial \bhb^*}\right]^\herm
\, \Big | \, \bhb \right\} \: .
\label{eq:app1_3}
\ee
Remembering that
$\bhb \sim
{\cal CN}(\mathbf{0}_{L_\text{B}},\bI_{L_\text{B}})$ by assumption, it results that
 ${\partial \ln p(\bhb)}/{\partial \bhb^*}=\bhb$, thus
 yielding
\be
\Fcal=\Es_{\bhb}[\Fcal_{\text{c}}] + \bI_{L_\text{B}} \: .
\label{eq:app1_4}
\ee
On the other hand, since $\bd \sim
{\cal CN}(\mathbf{0}_{\Nr}, \bR_{\bd\bd})$, one gets
\be
\frac{\partial \ln p(\by \, | \, \bhb)}{\partial \bhb^*} =
- \bK_\text{B}^\herm \, \bR_{\bd\bd}^{-1} \, \by +
\bK_\text{B}^\herm \, \bR_{\bd\bd}^{-1} \,  \bK_\text{B} \, \bhb \: .
\label{eq:app1_6}
\ee
Accounting for \eqref{eq:app1_6}, the matrix $\Fcal_{\text{c}}$ defined in
\eqref{eq:app1_3} can be expressed as
$\Fcal_{\text{c}} = \bK_\text{B}^\herm \, \bR_{\bd\bd}^{-1} \,  \bK_\text{B}$,
which does not depend on $\bhb$. Therefore, owing to \eqref{eq:app1_4},
it results that $\text{BCRLB}(\bhb)$ exactly coincides with
\eqref{eq:BMMSE-1}.

\section{Proof of Lemma~\ref{lem:3}}
\label{app:lem-3}

By virtue of the matrix inversion lemma \cite{Horn.book.1991}, one has
\begin{multline}
\left(\bI_{L_\text{B}}+\bK_\text{B}^\herm \, \bR_{\bd\bd}^{-1} \,  \bK_\text{B}\right)^{-1}  = \bI_{L_\text{B}} -
\bK_\text{B}^\herm \left(\bK_\text{B} \, \bK_\text{B}^\herm +
\bR_{\bd\bd} \right) \bK_\text{B}
\\  =
\bI_{L_\text{B}} -
\bJ^\trasp \bK \left(\bK \, \bK^\herm +
\sigma_v^2 \, \bI_{\Nr}\right) \bK \, \bJ
\label{eq:app7}
\end{multline}
where we have also used the expression of $\bR_{\bd\bd}$ and defined
$\bJ \eqdef [\bI_{L_\text{B}}, \bO_{L_\text{B} \times L_\text{E}}]^\trasp$
and $\bK \eqdef [\bK_\text{B}, \bK_\text{E}]
\in \C^{\Nr \times (L_\text{B}+L_\text{E})}$.
By substituting \eqref{eq:app7} in \eqref{eq:BMMSE-1} and
using the limit formula for the Moore–Penrose inverse \cite{Ben.book.2002}, one has
\be
\overline{\text{BMSE}}_\text{MMSE}
= L_\text{B} - \trace\left[\bJ^\trasp \left( \bK^\dag  \bK\right)^\herm \bJ\right]
\ee
from which follows that $\overline{\text{BMSE}}_\text{MMSE}=0$ if $\bK$ has full-column rank,
i.e., $\bK^\dag  \bK=\bI_{L_\text{B}}$. The proof is completed by observing that
$\bK$ is full column rank if and only if the columns of
$\bA_\text{B}$ and $\bA_\text{E}$ are
linearly independent.

\section{Truncated distributions}
\label{app:trunc}

A real-valued random variable
$X$ is said to follow a \textit{truncated Gaussian distribution} over the
interval $[a, b]$ (with $b>a$) -- denoted as
$X \sim {\cal N}_\text{T}(\mu, \sigma, a,b)$ --
if its pdf  is given by (see, e.g., \cite{Johnson})
\be
p_X(x)=
\begin{cases}
\frac{1}{C \sqrt{2 \pi} \sigma}
\, e^{- \frac{(x-\mu)^2}{2 \sigma^2}} \: ,
& \text{for $a \le x \le  b$} \: , \\
0 \: ,& \text{otherwise} \: ,
\end{cases}
\ee
where $\mu$ and $\sigma$ are the mean and
standard deviation of the corresponding
untruncated Gaussian distribution, respectively, and
$C \eqdef \left\{\text{erf}\left[(b-\mu)/{(\sqrt{2} \, \sigma)}\right]
- \text{erf}\left[(a-\mu)/{(\sqrt{2} \, \sigma)}\right]\right\}/2$
ensures that $p_X(x)$ is a valid pdf.
It is important to emphasize that $\mu$ and $\sigma$ are shape parameters for
the truncated distribution - they are not the mean and the standard deviation
of $X$. In the special case of $a=-b<0$ and $\mu=0$, it follows \cite{Johnson} that
$\Es[X]=0$ and
$\Es[X^2]= \sigma^2 \left[ 1- 2 \, b \, p_X(b) \right]$.

Let $X \sim {\cal N}_\text{T}(\mu, \sigma, a,b)$, the random variable
$Y=e^X$ exhibits a \textit{truncated lognormal distribution}, i.e., its pdf
can be written as
\be
p_Y(y)=
\begin{cases}
\frac{1}{C \sqrt{2 \pi} \sigma y}
\, e^{- \frac{[\ln(y)-\mu]^2}{2 \sigma^2}} \: ,
& \text{for $e^a \le y \le  e^b$} \: , \\
0 \: ,& \text{otherwise} \: .
\end{cases}
\ee
In the special case of $a=-b<0$ and $\mu=0$, one has
(details are omitted for the sake of brevity)
\begin{multline}
\Es[Y]=\Es[e^X]=\Es[e^{-X}]
\\ =\frac{e^{\frac{\sigma^2}{2}}}{2 \, \text{erf}\left(\frac{b}{\sqrt{2} \, \sigma}\right)} \left[ \text{erf}\left(\frac{b-\sigma^2}{\sqrt{2} \, \sigma}\right)
+ \text{erf}\left(\frac{b+\sigma^2}{\sqrt{2} \, \sigma}\right) \right] \: .
\label{app:boh}
\end{multline}
It is noteworthy that $\Es[Y] =1$, as $\sigma \to 0$.

\section{Case 2: Derivation of the improved LMMSEE}
\label{app:lmmsee}

The general expression of the LMMSEE \cite{Kay} is given by
\be
\bhbhatlmmsedue = \Es \left[ \bhb \, \by^\herm \right] \, \left(\Es \left[ \by \, \by^\herm \right]\right)^{-1} \by
\label{eq:app8}
\ee
where the expectation is taken not only over
$(\bhb, \bh_\text{E},\bv)$, but over
the probability densities of
$\Delta\theta_{\text{E},1}, \Delta\theta_{\text{E},2}, \ldots, \Delta\theta_{\text{E},L_\text{E}}$, and $\Delta\Ps$ as well.
It is readily seen that $\Es \left[ \bhb \, \by^\herm \right]=\bK_\text{B}^\herm $.

By resorting to the conditional expectation rule \cite{Casella}, the
correlation matrix in \eqref{eq:app8} can be equivalently written as
\barr
\Es \left[ \by \, \by^\herm \right] &= \Es_{
\{\Delta\theta_{\text{E},\ell}\}_{\ell=1}^{L_\text{E}},
\Delta\Ps}
\left[ \bR_{\by\by} \right]
\nonumber \\ & = \bK_\text{B} \, \bK_\text{B}^\herm +
\Es_{
\{\Delta\theta_{\text{E},\ell}\}_{\ell=1}^{L_\text{E}},
\Delta\Ps}
\left[ \bK_\text{E} \, \bK_\text{E}^\herm \right] + \sigma_v^2 \, \bI_{\Nr}
\label{eq:app9}
\earr
where the (conditional) correlation matrix $\bR_{\by\by}$, given $\bK_\text{E}$,
has been defined in \eqref{eq:Ryy}.
According to \eqref{eq:dirhat} and \eqref{eq:powhat},
we remember that $\bK_\text{E} =
\sqrt{\Ps} \bA_\text{E}$, with
$\Ps  = \Pshat \, e^{-\Delta \Ps}$ and
\begin{multline}
\bA_\text{E}  =\frac{1}{\sqrt{L_\text{E}}}
\left[ \ba(\widehat{\theta}_{\text{E},1}-\Delta \theta_{\text{E},1}),
\ba(\widehat{\theta}_{\text{E},2}-\Delta \theta_{\text{E},2}),
\right. \\ \left.
\ldots,
\ba(\widehat{\theta}_{\text{E},L_\text{E}}-\Delta \theta_{\text{E},L_\text{E}})\right] \: .
\label{eq:app10}
\end{multline}
Using the properties of the Kronecker product, one gets
\begin{multline}
\Es_{
\{\Delta\theta_{\text{E},\ell}\}_{\ell=1}^{L_\text{E}},
\Delta\Ps}
\left[ \bK_\text{E} \, \bK_\text{E}^\herm \right]
\\ = \Pshat \, \Es\left[e^{-\Delta \Ps}\right] 
\Es_{
\{\Delta\theta_{\text{E},\ell}\}_{\ell=1}^{L_\text{E}}} \left[ \bA_\text{E}  \,
\bA_\text{E}^\herm \right]
\\ = \Pshat \, \Es\left[e^{-\Delta \Ps}\right] 
\frac{1}{L_\text{E}} \sum_{\ell=1}^{L_\text{E}}
\bR_{\ba\ba}^{(\ell)}
\label{eq:app11}
\end{multline}
where we have defined
\be
\bR_{\ba\ba}^{(\ell)} \eqdef \Es_{
\Delta\theta_{\text{E},\ell}} \left[\ba(\widehat{\theta}_{\text{E},\ell}-\Delta \theta_{\text{E},\ell}) \, \ba^\herm(\widehat{\theta}_{\text{E},\ell}-\Delta \theta_{\text{E},\ell})\right] \: .
\label{eq:app12}
\ee
By virtue of \eqref{app:boh}, it is readily seen that
$\Es\left[e^{-\Delta \Ps}\right]$ is given by
\eqref{eq:meandeltaPE}.
The $(n_1+1,n_2+1)$th entry of $\bR_{\ba\ba}^{(\ell)}$ is given by
\begin{multline}
\left\{\bR_{\ba\ba}^{(\ell)}\right \}_{n_1+1,n_2+1}= \frac{\Es_{
\Delta\theta_{\text{E},\ell}} \left[ e^{-j 2 \pi (n_1-n_2) \Delta \cos(
\widehat{\theta}_{\text{E},\ell}-\Delta \theta_{\text{E},\ell})}\right]
}{N} \\ =
\frac{\Es_{\Delta\theta_{\text{E},\ell}}
\left\{\cos\left[2 \pi (n_1-n_2) \Delta
\cos(\widehat{\theta}_{\text{E},\ell}-\Delta \theta_{\text{E},\ell})\right]\right\}}{N}
\\
\frac{-j \, \Es_{\Delta\theta_{\text{E},\ell}}
\left\{\sin\left[2 \pi (n_1-n_2) \Delta
\cos(\widehat{\theta}_{\text{E},\ell}-\Delta \theta_{\text{E},\ell})\right]\right\}}{N}
\label{eq:app13}
\end{multline}
for $n_1,n_2 \in \{0,1,\ldots, N-1\}$.  By resorting to the
difference formula for cosine, one gets
\begin{multline}
\cos(
\widehat{\theta}_{\text{E},\ell}-\Delta \theta_{\text{E},\ell})=
\cos(
\widehat{\theta}_{\text{E},\ell}) \, \cos(\Delta \theta_{\text{E},\ell})
\\ +
\sin(
\widehat{\theta}_{\text{E},\ell}) \, \sin(\Delta \theta_{\text{E},\ell})
\\ \approx \cos(
\widehat{\theta}_{\text{E},\ell}) \, \left[1- \frac{\left(\Delta \theta_{\text{E},\ell}\right)^2}{2}\right]
+ \sin(
\widehat{\theta}_{\text{E},\ell}) \,  \Delta \theta_{\text{E},\ell}
\label{eq:app14}
\end{multline}
where we have also approximated $\cos(\Delta \theta_{\text{E},\ell})$
and $\sin(\Delta \theta_{\text{E},\ell})$ by using their corresponding
second-order Maclaurin series expansion, under the assumption that
the estimation error is sufficiently small.
By substituting \eqref{eq:app14} in \eqref{eq:app13}, employing
the difference formulas for cosine and sine, taking
second-order Maclaurin series expansion
of $\cos[\beta_\text{I}(n_1,n_2,\ell) \,
{\left(\Delta \theta_{\text{E},\ell}\right)^2}/{2}]$,
$\cos[\beta_\text{Q}(n_1,n_2, \ell) \,
\Delta \theta_{\text{E},\ell}]$,
$\sin[\beta_\text{I}(n_1,n_2,\ell) \,
{\left(\Delta \theta_{\text{E},\ell}\right)^2}/{2}]$,
$\sin[\beta_\text{Q}(n_1,n_2,\ell) \,
\Delta \theta_{\text{E},\ell}]$, and neglecting
all the terms that tend to zero faster than
$\left(\Delta \theta_{\text{E},\ell}\right)^2$, one has the approximations
\begin{multline}
\cos\left[2 \pi (n_1-n_2) \Delta
\cos(\widehat{\theta}_{\text{E},\ell}-\Delta \theta_{\text{E},\ell})\right]
\approx
\cos\left[\beta_\text{I}(n_1,n_2,\ell)\right]
\\ \cdot
\left[ 1- \beta_\text{Q}^2(n_1,n_2,\ell) \,
\frac{\left(\Delta \theta_{\text{E},\ell}\right)^2}{2}
\right] +
\sin\left[\beta_\text{I}(n_1,n_2,\ell)\right]
\\ \cdot
\left[\beta_\text{I}(n_1,n_2,\ell) \,
\frac{\left(\Delta \theta_{\text{E},\ell}\right)^2}{2}
- \beta_\text{Q}(n_1,n_2,\ell) \, \Delta \theta_{\text{E},\ell}
\right]
\label{eq:app15}
\end{multline}
\begin{multline}
\sin\left[2 \pi (n_1-n_2) \Delta
\cos(\widehat{\theta}_{\text{E},\ell}-\Delta \theta_{\text{E},\ell})\right]
\approx
\sin\left[\beta_\text{I}(n_1,n_2,\ell)\right]
\\ \cdot
\left[ 1- \beta_\text{Q}^2(n_1,n_2,\ell) \,
\frac{\left(\Delta \theta_{\text{E},\ell}\right)^2}{2}
\right] -
\cos\left[\beta_\text{I}(n_1,n_2,\ell)\right]
\\ \cdot
\left[\beta_\text{I}(n_1,n_2,\ell) \,
\frac{\left(\Delta \theta_{\text{E},\ell}\right)^2}{2}
- \beta_\text{Q}(n_1,n_2,\ell) \, \Delta \theta_{\text{E},\ell}
\right]
\label{eq:app16}
\end{multline}
where
$\beta_\text{I}(n_1,n_2,\ell) \eqdef 2 \pi (n_1-n_2) \, \Delta \,
\cos(\widehat{\theta}_{\text{E},\ell})$ and
$\beta_\text{Q}(n_1,n_2,\ell) \eqdef 2 \pi (n_1-n_2) \,\Delta \,
\sin(\widehat{\theta}_{\text{E},\ell})$.
Eq.~\eqref{eq:LMMSE} comes from taking the expectation of
\eqref{eq:app15} and \eqref{eq:app16} over the pdf of
$\Delta \theta_{\text{E},\ell}$ (see Appendix~\ref{app:trunc}) and substituting the results
in \eqref{eq:app13}.

%=======================Bibliography==============================%

%%%%%%%%%%%%%%%%%%%%%%%%%%%%%%%%%%%%%%%%%%%%

\begin{IEEEbiography}
[{\includegraphics[width=1in,height=1.25in,clip,keepaspectratio]{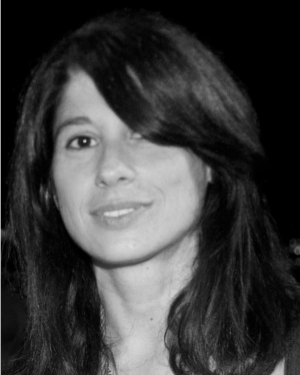}}]
{Donatella Darsena} (M'06-SM'16) received the Dr. Eng. degree summa cum laude in telecommunications engineering in 2001, and the Ph.D. degree in electronic and telecommunications engineering in 2005, both from the University of Napoli Federico II, Italy. From 2001 to 2002, she was an engineer in the Telecommunications, Peripherals and Automotive Group, STMicroelectronics, Milano, Italy. Since 2005, she has been an Assistant Professor with the Department of Engineering, University of Napoli Parthenope, Italy. Her research activities lie in the area of statistical signal processing, digital communications, and communication systems. In particular, her current interests are focused on equalization, channel identification, narrowband-interference suppression for multicarrier systems, space-time processing for cooperative communications systems and cognitive communications systems, and software-defined networks. Dr. Darsena has served as an Associate Editor for the IEEE COMMUNICATIONS LETTERS since December 2016, Senior Area Editor of the IEEE COMMUNICATIONS LETTERS since 2019 and Associate Editor for IEEE Access since October 2018
\end{IEEEbiography}

\vspace*{-2\baselineskip}

\begin{IEEEbiography}
[{\includegraphics[width=1in,height=1.25in,clip,keepaspectratio]{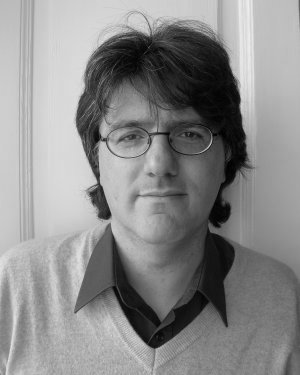}}]
{Giacinto Gelli}(M'18, SM'20)
was born in Napoli, Italy, on July 29, 1964.
He received the Dr. Eng. degree \textit{summa cum laude} in electronic
engineering in 1990, and the Ph.D. degree in computer science and
electronic engineering in 1994, both from the University of Napoli
Federico II.

From 1994 to 1998, he was an Assistant Professor with the
Department of Information Engineering, Second University of
Napoli.
Since 1998 he has been with the Department of Electrical Engineering and Information Technology, University of Napoli Federico II,
first as an Associate Professor,
and since November 2006 as a Full Professor of Telecommunications.
He also held teaching positions at the University Parthenope of
Napoli.
His research interests are in the broad area of
signal and array processing for communications,
with current emphasis on backscatter communications,
multicarrier modulation systems and
space-time techniques for cooperative and cognitive
communications systems.
\end{IEEEbiography}

\vspace*{-2\baselineskip}

\begin{IEEEbiography}
[{\includegraphics[width=1in,height=1.25in,clip,keepaspectratio]{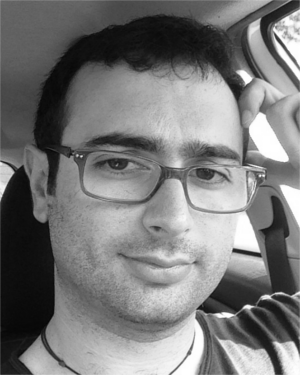}}]
{Ivan Iudice}
was born in Livorno, Italy, on November 23, 1986.
He received the B.S. and M.S. degrees
in telecommunications engineering in 2008 and 2010,
respectively, and the Ph.D. degree
in information technology and electrical engineering in 2017,
all from University of Napoli Federico II, Italy.

Since 2011, he has been part of
the Electronics and Communications laboratory
at Italian Aerospace Research Centre (CIRA),
Capua, Italy.
His research activities lie in the area of
signal and array processing for communications,
with current interests are focused
on physical layer cyber security
and space-time techniques for
cooperative communications systems.
\end{IEEEbiography}

\vspace*{-2\baselineskip}

\begin{IEEEbiography}[
{\includegraphics[width=1in,height=1.25in,clip,keepaspectratio]{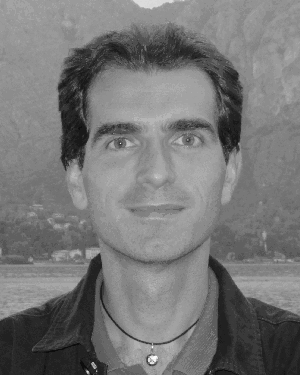}}]
{Francesco Verde}(M'10-SM'14) was born in Santa Maria Capua Vetere,
Italy, on June 12, 1974. He received the Dr. Eng. degree
\textit{summa cum laude} in electronic engineering
from the Second University of Napoli, Italy, in 1998, and the Ph.D.
degree in information engineering
from the University of Napoli Federico II, in 2002.
Since December 2002, he has been with the University of Napoli Federico II. He first served as an Assistant Professor of signal theory and mobile communications
and, since December 2011, he has served as an Associate Professor of telecommunications with the Department of Electrical Engineering and Information Technology.
His research activities include orthogonal/non-orthogonal multiple-access techniques, space-time processing for cooperative/cognitive communications, wireless systems optimization, and software-defined networks.

Prof. Verde has been involved in several technical program committees of major IEEE conferences in signal processing and wireless communications.
He has served as Associate Editor for IEEE TRANSACTIONS ON COMMUNICATIONS since 2017 and Senior Area Editor of the IEEE SIGNAL PROCESSING LETTERS since 2018. He was an Associate Editor of the IEEE TRANSACTIONS ON SIGNAL PROCESSING (from 2010 to 2014) and  IEEE SIGNAL PROCESSING LETTERS (from 2014 to 2018), as well as Guest Editor of the EURASIP Journal on Advances in Signal Processing in 2010 and SENSORS MDPI in 2018. 
\end{IEEEbiography}

\end{document}